\documentclass[times,sort&compress,3p,twocolumn]{elsarticle}
\journal{a journal for publication}
\usepackage[labelfont=bf]{caption}
\usepackage{bbm}
\usepackage{tikz}
\usepackage{multirow}

\usepackage{amsmath}
\usepackage{color}
\usepackage{mathrsfs}

\usetikzlibrary{automata,positioning}

\allowdisplaybreaks

\usepackage{subfigure}

\usepackage{amsmath,amsfonts,amssymb,amsthm,booktabs,color,epsfig,graphicx,hyperref,url}

\theoremstyle{plain}
\newtheorem{theorem}{Theorem}
\newtheorem{Ass}{Assumption}
\newtheorem{proposition}{Proposition}
\newtheorem{lemma}{Lemma}
\newtheorem{corollary}{Corollary}

\theoremstyle{definition}
\newtheorem{Def}{Definition}
\newtheorem{Rem}{Remark}

\begin{document}

\begin{frontmatter}

\title{Some results on maximum likelihood from incomplete data: finite sample properties and improved M-estimator for resampling}

\author[1]{Budhi Surya}

\address[1]{School of Mathematics and Statistics, Victoria University of Wellington,
Gate 6, Kelburn PDE, Wellington 6140, New Zealand}

\cortext[mycorrespondingauthor]{Corresponding author. Email address: \url{budhi.surya@vuw.ac.nz}}

\begin{abstract}
This paper presents some results on the maximum likelihood (ML) estimation from incomplete data. Finite sample properties of conditional observed information matrices are established. They possess positive definiteness and the same Loewner partial ordering as the expected information matrices do. An explicit form of the observed Fisher information (OFI) is derived for the calculation of standard errors of the ML estimates. It simplifies Louis (1982) general formula for the OFI matrix. To prevent from getting an incorrect inverse of the OFI matrix, which may be attributed by the lack of sparsity and large size of the matrix, a monotone convergent recursive equation for the inverse matrix is developed which in turn generalizes the algorithm of Hero and Fessler (1994) for the Cram\'er-Rao lower bound. To improve the estimation, in particular when applying repeated sampling to incomplete data, a robust M-estimator is introduced. A closed form sandwich estimator of covariance matrix is proposed to provide the standard errors of the M-estimator. By the resulting loss of information presented in finite-sample incomplete data, the sandwich estimator produces smaller standard errors for the M-estimator than the ML estimates. In the case of complete information or absence of re-sampling, the M-estimator coincides with the ML estimates. Application to parameter estimation of a regime switching conditional Markov jump process is discussed to verify the results. The simulation study confirms the accuracy and asymptotic properties of the M-estimator.\end{abstract}


\end{frontmatter}

\section{Introduction}

Large sample properties of maximum likelihood estimates (MLE) of statistical model parameters have been well documented in literature. See for e.g. Cram\'er (1946), Newey and McFadden (1994), and Van der Vaart (2000). The properties state that as the sample size increases, the MLE has asymptotic (multivariate) normal distribution with mean equal to the true parameter value whereas the covariance matrix is given by the inverse expected Fisher information of observed data. These fine properties of MLE were first shown by Fisher (1925) and later established rigorously, among others, by Cram\'er (1946). For unbiased estimators, the information matrix corresponds to the Cram\'er-Rao lower bound. See for e.g. p.489 of Cram\'er (1946) or p.2162 of Newey and McFadden (1994). For finite sample application, the observed Fisher information has been widely used to find the standard errors of the MLE with good accuracy (Efron and Hinkley, 1978). Although the large sample properties have been well developed for complete data, some further studies are required to understand the finite/large sample properties of the MLE when applied to an incomplete data. 
 
The EM algorithm developed by Dempster et al. (1977) for maximum likelihood estimation under incomplete data has widely been used over various fields in literature. It is a robust and powerful tool for statistical analysis with missing data (Little and Rubin, 2020). The algorithm provides an iterative approach to obtain the maximum likelihood estimates of the model parameters in a way that avoids necessary regularity conditions on the log-likelihood function in terms of the existence of its second derivative and invertibility of the corresponding Hessian matrix. Otherwise, if the regularity conditions are satisfied for each observation, one may employ the Newton-Raphson approach, i.e., the Fisher scoring method to find the MLE also iteratively, see for e.g. Osborne (1992), Hastie et al. (2009), and Takai (2020). There are two steps in the EM estimation. The first step, the \textit{E-step}, involves valuation under current parameter estimate of conditional expectation of the log-likelihood given the observed data, whereas the second step, the \textit{M-step}, deals with optimizing the conditional expectation. In each iteration, the algorithm increases the value of log-likelihood. The appealing monotone convergence property of the observed-data log-likelihood gives a higher degree of stability for the convergence of the EM algorithm (Wu, 1983). We refer to McLachlan and Krishnan (2008) for recent developments, extensions and applications of the EM algorithm. 

However, the EM algorithm only provides point estimates of parameters. Unlike the Fisher scoring method, it does not automatically produce the covariance matrix of the MLE. Additional steps are required to find the covariance matrix. For finite and independent data, the expected information is replaced by the observed information matrix specified by the second derivative of the observed data log-likelihood function (Efron and Hinkley, 1978). 

In general, the second derivative is very difficult to evaluate directly. One of major contribution on the evaluation of the observed information was given by Louis (1982) in which a general matrix formula was proposed. Notice that the Louis' formula involves conditional expectation of outer product of the complete-data score function which in general may be complicated to simplify. Meng and Rubin (1991) derived the covariance matrix using the fundamental identity given in Dempster et al. (1977) relating the log-likelihood of observed data, the EM-update criterion function and the conditional expectation of conditional log-likelihood of complete data given its incomplete observation. In working out the information matrix, they applied in the M-step of the EM-algorithm a first-order Taylor approximation around current parameter estimate to the EM-update function. This approximation was first noted in Meilijson (1989) in an attempt to provide a fast improvement to the EM algorithm. The covariance matrix derived in Meng and Rubin (1991) is given by the difference between the complete information matrix and an incomplete one, although in a slightly more complex form than Louis' formula. Based on the Taylor approximation discussed in Meilijson (1989), Jamshidian and Jennrich (2000) proposed a numerical differentiation method to evaluate the covariance matrix for the incomplete data. A rather direct calculation of observed information matrix was discussed in Oakes (1999) based on taking derivatives of the fundamental EM identity. 

All above methods provide convergent and consistent MLE of the true parameter whose consistent limiting normal distribution has the mean equal to the true value, whereas the asymptotic covariance matrix is specified by the inverse expected Fisher information of the observed data. As a result, the covariance matrix estimator is determined by the inverse of observed Fisher information.

However, inverting the observed Fisher information matrix may be problematic due to lack of sparsity and the (large) size of the information matrix. This might result in getting an incorrect inverse of the matrix. To overcome this problem, a monotone convergent iterative scheme is proposed to calculate the inverse. Furthermore, to improve the parameter estimation, in particular when applying repeated sampling such as the Bootstrapping method, see Efron and Hastie (2021), this paper proposes a robust M-estimator with smaller standard errors than that of the corresponding MLE. A closed form sandwich estimator of finite-sample covariance matrix is introduced to provide the standard errors of the M-estimator. The sandwich estimator reserves itself as the lower bound for the inverse of the observed Fisher information. It is slightly different from Huber sandwich estimator (Huber, 1967; Freedman, 2006; Little and Rubin, 2020) for model misspecification under incomplete data. In the absence of repeated sampling, the M-estimator coincides with the MLE.

This paper is organized as follows. Section \ref{sec:sec1} discusses maximum likelihood estimation from incomplete data and respective problems. Section \ref{sec:main} presents the main results and contributions of this paper. An example on conditional Markov jump processes is presented in Section \ref{sec:sec4}. A series of simulation studies based on the example of Section \ref{sec:sec4} are performed in Section \ref{sec:sec5} to verify the main results of Section \ref{sec:main}. Section \ref{sec:sec6} concludes this paper. 

\section{Maximum likelihood estimation from incomplete data}\label{sec:sec1}

\subsection{The likelihood function of incomplete data}
Let $X$ and $Y$ be two random vectors defined on the same probability space $(\Omega,\mathscr{F},\mathbb{P})$. Denote by $(\mathcal{X},\mathcal{S})$ and $(\mathcal{Y},\mathcal{T})$ the corresponding measurable state spaces of $X$ and $Y$ and by $T:\mathcal{X}\rightarrow \mathcal{Y}$ a many-to-one mapping from $\mathcal{X}$ to $\mathcal{Y}$. Suppose that a \textit{complete-data} vector $x\in\mathcal{X}$ is partially observed through an \textit{incomplete-data} vector $y=T(x)$ in $\mathcal{Y}$. Assume that there exist probability density functions $f_c(x\vert\theta)$ and $f_o(y\vert\theta)$ corresponding to the complete data $x\in\mathcal{X}$ and its incomplete observation $y\in\mathcal{Y}$, respectively. Here $\theta$ represents a vector of parameters on a parameter space $\Theta$, with $d=\vert \theta\vert$, characterizing the distribution of $X$. For formality, we assume that $\Theta$ is a compact set. Define $\mathcal{X}(y)=\{x\in\mathcal{X}: T(x)=y\}\in \mathcal{S}$. Then,
\begin{align}\label{eq:cond}
f_o(y\vert \theta)=\int_{\mathcal{X}(y)} f_c(x\vert\theta)\lambda(dx),
\end{align}
where $\lambda$ is a $\sigma-$finite measure on $\mathcal{S}$, absolutely continuous with respect to the probability distribution $\mathbb{P}\circ X^{-1}$ with the density function $f_c(x\vert \theta)$ (the Radon-Nikodym derivative). See e.g. Halmos and Savage (1949) for details. Notice that our description of the marginal distribution (\ref{eq:cond}) for incomplete data is slightly more general than the one employed in McLachlan and Krishnan (2008).

Following the identity (\ref{eq:cond}), the conditional probability density function $f(x\vert y,\theta)$ of the complete-data $X$ given its incomplete observation $Y$ is therefore given by
\begin{align}\label{eq:condpdf}
f(x\vert y,\theta)=\frac{f_c(x\vert\theta)}{f_o(y\vert \theta)}.
\end{align}

\begin{Ass}\label{ass:A12}
Without loss of generality, we assume throughout the remaining of this paper that the log-likelihood function $\log f_c(x\vert\theta)$ is twice continuously differentiable w.r.t $\theta$ and for all $\theta\in\Theta$, $m\in\{0,1,2\},$
\begin{align}\label{eq:bounded}
\int_{\mathcal{X}(y)} \Big\vert \frac{\partial^m \log f_c(x\vert\theta) }{\partial \theta^m}\Big\vert f(x\vert y,\theta)\lambda(dx)<\infty. \tag{A1}
\end{align}
\end{Ass}
\noindent Condition (\ref{eq:bounded}) verifies the existence of expectation $\mathbb{E}\big[\big\vert \frac{\partial^m \log f_c(X\vert\theta) }{\partial \theta^m}  \big\vert \big\vert Y=y,\theta^0\big]$  for all $\theta\in\Theta$, $m\in\{0,1,2\}$. 

Suppose that a complete data $X_1,\ldots, X_n$ were generated independently from the density function $f_c(x\vert\theta)$ under a pre-specified parameter value $\theta^0\in\Theta$. We assume throughout that each outcome $X^k$, $k\in\{1,\ldots, n\},$ is partially observed, represented by an incomplete-data vector $Y^k=T(X^k)$. And each observation $Y^k$ of $X^k$ is also independent. Namely, if $X^k$ is independent of $X^{\ell}$, for $k\neq \ell$, then $Y^k$ is independent of $Y^{\ell}$. The log-likelihood contribution of the incomplete observation $Y=\bigcup_{k=1}^n Y^k$,
\begin{align*}
\log f_o(Y\vert \theta)=\sum_{k=1}^n \log \int_{\mathcal{X}(y^k)} f_c(x^k\vert\theta)\lambda(dx^k),
\end{align*}
is used to get an estimator $\widehat{\theta}_n$ of $\theta^0$ defined as the global maximizer $\textrm{Argmax}_{\theta\in\Theta^d} \log f_o(Y\vert \theta)$ of the log-likelihood function. This method proves the convergence of $\widehat{\theta}_n$ to the true value $\theta^0$ as the sample size $n$ increases, where the convergence occurs almost surely under $\mathbb{P}_{\vert\theta^0}:=\mathbb{P}\{\bullet\vert\theta^0\}$. 

\subsection{M-criterion for maximum likelihood estimation}\label{sec:sec22}
The following proposition establishes the above claim.
\begin{proposition}\label{prop:prop1}
By independence of the observations $\{Y^k\}$, $\widehat{\theta}_n\overset{\mathbb{P}_{\vert \theta^0}%
}{\Longrightarrow}\theta^0$ as the sample size $n$ increases.
\end{proposition}
\begin{proof}
To prove the claim, consider the \textbf{M-criterion} 
\begin{align}\label{eq:criteriamle}
M_n(\theta)=\frac{1}{n}\sum_{k=1}^n \log f_o(Y^k\vert\theta).
\end{align}
See Van der Vaart (2000). By independence of $\{Y^k\}$, $M_n(\theta)\overset{\mathbb{P}_{\vert \theta^0}}{\Longrightarrow}M(\theta)=\mathbb{E}\big[\log f_o(Y^k\vert\theta)\vert \theta_0\big]$ which holds for a generic sample $Y^k$. Furthermore, from the Shanon-Kolmogorov information inequality, it holds for any $\theta \neq \theta ^{0}$ and generic $Y^{k}$, $R(\theta ^{0},\theta ):=\mathbb{E}\big[\log \big(\frac{f_o(Y^k\vert\theta^{0})}{f_o(Y^{k}\vert \theta )}\big)\big\vert \theta^0\big]>0,$ see p. 113 in Ferguson (1996). Thus, $\sup_{\theta \in
\Theta\backslash \theta^0} M(\theta) < M(\theta^0) \iff\theta^0=\text{arg}%
\max_{\theta\in\Theta} M(\theta).$ Since $\widehat{\theta}_n$ is the global
maximizer of $M_n(\theta)$ and the latter converges with probability one to $%
M(\theta)$, it follows that $\widehat{\theta}_n$ gets closer and closer to the
global maximizer $\theta^0$ of $M(\theta)$ as $n$ increases,
which by compactness of $\Theta$ implies that $\widehat{\theta}_n \overset{\mathbb{P}_{\vert \theta^0}}{\Longrightarrow} \theta^0$.
\end{proof}

By (\ref{eq:bounded}) and the first order Euler condition, the MLE $\widehat{\theta}_n$ is found as the solution of the systems of equation 
\begin{align}\label{eq:score}
0=S_n(\theta):=\frac{1}{n}\sum_{k=1}^n \frac{\partial \log f_o(Y^k\vert \theta)}{\partial \theta}.
\end{align}
Assuming continuous differentiability of the score function $\frac{\partial \log f_o(Y^k\vert \theta)}{\partial \theta}$, for large $n$, the consistency of $\widehat{\theta}_n$ allows one to apply the first-order Taylor approximation around $\theta^0$ to $S_n(\widehat{\theta}_n)$, see e.g. Freedman (2006), to arrive at
\begin{align}\label{eq:pers1}
\widehat{\theta}_n=&\theta^0 + J_y^{-1}(\theta^0)S_n(\theta^0),
\end{align}%
where $J_y^{-1}(\theta)$ is the observed Fisher information given by
\begin{align}\label{eq:OFIM}
J_y(\theta)=-\frac{1}{n}\sum_{k=1}^n \frac{\partial^2 f_o(Y^k\vert\theta)}{\partial \theta^2}.
\end{align}

By independence of $\{Y^k\}$, application of the Slutsky's lemma (see, Lemma 2.8 in Van der Vaart (2000)) and the Central Limit Theorem (CLT), it is known that $\sqrt{n}S_n(\theta^0)\sim N(0, I_y(\theta^0))$, $I_y(\theta)=\mathbb{E}\big[-\frac{\partial^2 \log f_0(Y\vert\theta)}{\partial \theta^2}\big\vert\theta\big]$ being the expected Fisher information matrix. See, e.g. Cram\'er (1946), Newey and McFadden (1994) and Van der Vaart (2000). As a result, it follows from (\ref{eq:pers1}) that the MLE $\widehat{\theta}_n$ has the $\sqrt{n}-$asymptotic normal distribution $$\sqrt{n}(\widehat{\theta}_n-\theta^0)\sim N(0, I_y^{-1}(\theta^0)).$$

\subsection{Some estimation difficulties for incomplete data}

However, for incomplete data, it is generally difficult to obtain the MLE $\widehat{\theta}_n$ explicitly. There are some difficulties:
\begin{enumerate}
\item[(i)] $\widehat{\theta}_n$ solves the systems of (nonlinear) equations $$0=S_n(\widehat{\theta}_n)=\sum_{k=1}^n \frac{\int_{\mathcal{X}(y^k)}\frac{\partial f_c(x^k\vert \widehat{\theta}_n)}{\partial \theta}\lambda(dx^k)}{\int_{\mathcal{X}(y^k)} f_c(x^k\vert\widehat{\theta}_n)\lambda(dx^k)},$$ from which it is difficult to pull out $\widehat{\theta}_n$ explicitly.

\item[(ii)] Since $\theta^0$ is unknown, the estimator (\ref{eq:pers1}) is not applicable. Hence, $\widehat{\theta}_n$ should be derived recursively. 

\item[(iii)] The observed Fisher information $J_y(\theta)$ and the score function $S_n(\theta)$ are not explicit. Therefore, the recursive Fisher scoring method (see, Hastie et al. (2009)) derived from (\ref{eq:pers1}) is difficult to implement.  

\item[(iv)] Although Louis (1982) formula is available for evaluating the observed Fisher information $J_y(\theta)$, it does not simplify the conditional expectation of outer product of the complete-data score function.

\item[(v)] In general the observed information matrix $J_y(\theta)$ lacks of sparsity and the size of $\theta$ could be large. Hence, inverting $J_y(\theta)$ directly might be difficult and may result in an incorrect inverse $J_y^{-1}(\theta)$.

\item[(vi)] Due to the resulting loss of information presented in incomplete data $Y$, the standard errors of the MLE $\widehat{\theta}_n$ is larger than that of using complete data $X$. 

\item[(vii)] The EM algorithm may be slow in its convergence and might be difficult to get an explicit iteration form due to possible nonlinearity of the likelihood function in terms of the model parameter.
\end{enumerate}

This paper attempts to solve the above problems and proposes improved estimation of the true value $\theta^0$ in terms of smaller standard error of its estimator, in particular, by applying repeated sampling method. 

\section{Main contributions}\label{sec:main}

To overcome the difficulties (i)-(vi), we consider estimating the true value $\theta^0$ based on the incomplete observation $\{Y^k, k=1,\ldots,n\}$ using the \textbf{M-criterion}
\begin{eqnarray}\label{eq:EM}
\mathscr{M}_n(\theta)=\frac{1}{n}\sum_{k=1}^n \mathbb{E}\Big[\log f_c(X^k\vert\theta)\Big\vert Y^k,\theta^0\Big], 
\end{eqnarray}
where $\mathbb{E}[\bullet \vert\theta^0]$ refers to the expectation operator associated with the underlying probability measure $\mathbb{P}\{\bullet \vert\theta^0\}$ from which the complete-data $\{X^k: 1\leq k\leq n\}$ were generated under $\theta^0$. Note that such M-criterion function as $\mathscr{M}_n(\theta)$ (\ref{eq:EM}) is not discussed in Van der Vaart (2000). 

An estimator $\widehat{\theta}_n^0$ of $\theta^0$ is defined as the maximizer of the M-criterion $\mathscr{M}_n(\theta)$ \eqref{eq:EM} over $\theta\in\Theta$, as the solution of
\begin{align}\label{eq:solsol}
0=\mathscr{S}_n(\theta):=\frac{1}{n}\sum_{k=1}^n \mathbb{E}\Big[\frac{\partial \log f_c(X^k\vert\theta)}{\partial \theta} \Big\vert Y^k,\theta^0\Big]. 
\end{align}
Define $\mathscr{S}(\theta):=\mathbb{E}\big[\frac{\partial \log f_c(X\vert \theta)}{\partial \theta}\big\vert \theta^0\big]$. It is clear that $\mathscr{S}(\theta^0)=0$. Thus, the M-estimator $\widehat{\theta}_n^0$ can simply be written as 
\begin{align}\label{eq:EMLE}
\widehat{\theta}_n^0=\mathscr{S}_n^{-1}\big(\mathscr{S}(\theta^0)\big).
\end{align}
As the M-estimator, we may apply similar arguments to the proof of Proposition \ref{prop:prop1} and the following result to establish the consistency of the M-estimator $\widehat{\theta}_n^0$ (\ref{eq:EMLE}). 
\begin{lemma}\label{lem:lem1}
Let $\mathscr{M}(\theta)=\mathbb{E}\big[\log f_c(X\vert\theta)\big\vert \theta^0\big]$. Then, for $\theta^0\in\Theta$, $\sup_{\theta \in
\Theta\backslash \theta^0}\mathscr{M}(\theta)\leq \mathscr{M}(\theta^0)$ with $\mathscr{M}^{\prime}(\theta^0)=0$. 
\end{lemma}
\begin{proof}
The proof follows from concavity of log-function and application of Jensen's inequality, i.e.,
\begin{align*}
\mathscr{M}(\theta)-\mathscr{M}(\theta^0)\leq \log\mathbb{E}\Big[\frac{f_c(X\vert\theta)}{f_c(X\vert\theta^0)}\Big\vert \theta^0\Big]= 0,
\end{align*}
where the last equality is due to the fact that the likelihood ratio $\frac{f_c(X\vert\theta)}{f_c(X\vert\theta^0)}$ corresponds to the Radon-Nikodym derivative of changing the probability measure from $\mathbb{P}\{\bullet \vert\theta^0\}$ to $\mathbb{P}\{\bullet \vert\theta\}$ which results in $\mathbb{E}\big[ \frac{f_c(X\vert\theta)}{f_c(X\vert\theta^0)} \big\vert \theta^0\big]=1$. $\mathscr{M}^{\prime}(\theta^0)=0$ follows from \eqref{eq:bounded} and that $\mathbb{E}\big[\frac{\partial \log f_c(X\vert \theta)}{\partial \theta}\big\vert \theta\big]=0$.  
\end{proof}
Alternatively, by independence of $\{Y^k\}$, the law of large number shows that $\mathscr{S}_n(\theta)\stackrel{\mathbb{P}}{\Longrightarrow}\mathscr{S}(\theta)$ for any $\theta\in\Theta$. As the result, $\widehat{\theta}_n^0\stackrel{\mathbb{P}_{\vert\theta^0}}{\Longrightarrow}\theta^0$ and we have asymptotically that $\theta\approx\mathscr{S}_n^{-1}(\mathscr{S}(\theta))$. By consistency of $\widehat{\theta}_n^0$ and the regularity condition (\ref{eq:bounded}), we may apply as before the first order Taylor approximation around $\theta^0$ to $\mathscr{S}_n(\widehat{\theta}_n^0)$ to obtain
\begin{align}\label{eq:EMLEb}
\widehat{\theta}_n^0=\theta^0 + J_x^{-1}(\theta^0)\mathscr{S}_n(\theta^0),
\end{align}
where $J_x(\theta)$ defines the conditional observed information
\begin{align*}
J_x(\theta)=\frac{1}{n}\sum_{k=1}^n \mathbb{E}\Big[-\frac{\partial^2 \log f_c(X^k\vert\theta)}{\partial \theta^2}\Big\vert Y^k, \theta\Big],
\end{align*}
which by independence of $\{Y^k\}$ converges to the complete-data expected information $I_x(\theta)= \mathbb{E}\big[-\frac{\partial^2 \log f_c(X\vert\theta)}{\partial \theta^2}\big\vert \theta\big].$

In the section below, finite sample properties of the information matrices $J_y(\theta)$ and $J_x(\theta)$ are studied. They will be used to derive recursive estimations for the MLE $\widehat{\theta}_n$ and for the Cr\'armer-Rao lower bound $J_y^{-1}(\widehat{\theta}_n)$ for the finite-sample covariance matrix of the MLE. In particular, to obtain the limiting normal distribution of $\widehat{\theta}_n^0$ (\ref{eq:EMLEb}).

\subsection{Finite sample properties of information matrices}\label{sec:sec3}

The results below will be used to derive an explicit form of the observed information $J_y(\theta)$, to establish the inequality concerned with resulting loss of information in a finite-sample incomplete data, to obtain the $\sqrt{n}-$limiting normal distribution of the M-estimator $\widehat{\theta}_n^0$ (\ref{eq:EMLE}), and its asymptotic efficiency compared to the MLE $\widehat{\theta}_n$ (\ref{eq:pers1}). They have remained largely unexamined in literature on statistical analysis of incomplete data, see e.g., McLachlan and Krishnan (2008) and Little and Rubin (2020) for details.
\begin{proposition}\label{prop:prop2}
For any $\theta_i\in \theta$ and $Y\in\mathcal{Y}$, 
\begin{align}\label{eq:condscore2}
\mathbb{E}\Big[\frac{\partial \log f(X\vert Y,\theta)}{\partial \theta_i}\Big\vert Y,\theta\Big]=0.
\end{align}
\end{proposition}
\begin{proof}
Using the conditional density function \eqref{eq:condpdf}, 
\begin{align*}
&\mathbb{E}\Big[\frac{\partial \log f(X\vert Y,\theta)}{\partial \theta_i}\Big\vert Y=y,\theta\Big]\\
&\hspace{1cm}=
\int_{\mathcal{X}(y)}  \frac{\partial \log f(x\vert y,\theta)}{\partial \theta_i} f(x\vert y,\theta) \lambda(dx)\\
&\hspace{1cm}= \int_{\mathcal{X}(y)} \frac{\partial f(x\vert y,\theta)}{\partial \theta_i}  \lambda(dx),
\end{align*}
from which the claim follows on account of the regularity condition (\ref{eq:bounded}) along with the identities \eqref{eq:cond} and \eqref{eq:condpdf}.
\end{proof}
Applying (\ref{eq:condscore2}) to \eqref{eq:condpdf} after taking the logarithm,
\begin{eqnarray}\label{eq:score2}
\hspace{0.75cm}\frac{\partial \log f_c(X\vert\theta)}{\partial \theta_i}=\frac{\partial \log f_o(Y\vert\theta)}{\partial \theta_i}
+\frac{\partial \log f(X\vert Y,\theta)}{\partial \theta_i},
\end{eqnarray}
leads to the following identity, which corresponds to eqn. (3.44) in Section 3.7 of McLachlan and Krishnan (2008). 
\begin{corollary}\label{cor:mainprop}
For any $\theta_i\in\theta$ and $Y\in\mathcal{Y}$, 
\begin{align}\label{eq:id10}
\mathbb{E}\Big[ \frac{\partial \log f_c(X\vert\theta)}{\partial \theta_i}\Big\vert Y,\theta\Big]=\frac{\partial \log f_o(Y\vert\theta)}{\partial \theta_i}.
\end{align}
\end{corollary}

The identity (\ref{eq:id10}) identifies $\mathscr{S}_n(\theta^0)$ (\ref{eq:solsol}) as $S_n(\theta^0)$ (\ref{eq:score}). The result of Proposition \ref{prop:prop2} leads to the following fact.

\begin{corollary}\label{cor:zerocov}
For any $(\theta_i,\theta_j)\in\theta$ and $Y\in\mathcal{Y}$,
\begin{align*}
\mathrm{Cov}\Big(\frac{\partial \log f_o(Y\vert\theta)}{\partial \theta_i},\frac{\partial \log f(X\vert Y,\theta)}{\partial \theta_j}\Big\vert Y,\theta\Big)=0.
\end{align*}
\end{corollary}
The result above shows that $\frac{\partial \log f_o(Y\vert\theta)}{\partial \theta_i}$ and $\frac{\partial \log f(X\vert Y,\theta)}{\partial \theta_j}$ are conditionally independent given incomplete data $Y$.

To show the Loewner partial ordering of $J_y(\theta)$ and $J_x(\theta)$ and to derive an explicit form of the observed information matrix $J_y(\theta)$, the results below are required. 

\begin{theorem}\label{theo:theomain}
The conditional information matrix $$J_{x\vert y}(\theta_i,\theta_j):=\mathbb{E}\Big[-\frac{\partial^2 \log f(X\vert Y, \theta)}{\partial \theta_i\partial \theta_j} \Big\vert Y, \theta\Big]$$
is positive definite for any $\theta\in\Theta$ satisfying the equation
\begin{align*}
J_{x\vert y}(\theta_i,\theta_j) = \mathbb{E}\Big[\frac{\partial \log f(X\vert Y,\theta)}{\partial \theta_i}\frac{\partial \log f(X\vert Y,\theta)}{\partial \theta_j}\Big\vert Y,\theta\Big].
\end{align*}
\end{theorem}
\begin{proof} 
To prove the identity, by the chain rule the identity below holds for any $(\theta_i,\theta_j)\in \theta$ and $y\in\mathcal{Y}$,
\begin{align*}
&\Big[-\frac{\partial^2 \log f(x\vert y,\theta)}{\partial \theta_i \partial \theta_j}\Big] f(x\vert y,\theta)\\
&\hspace{1cm}=-\frac{\partial}{\partial \theta_i}\Big[\Big[\frac{\partial \log f(x\vert y,\theta)}{\partial \theta_j}\Big]f(x\vert y,\theta)\Big]\\
&\hspace{1.5cm}+\Big[\frac{\partial \log f(x\vert y,\theta)}{\partial \theta_i}\Big]\Big[\frac{\partial \log f(x\vert y,\theta)}{\partial \theta_j}\Big] f(x\vert y,\theta).
\end{align*} 
Therefore, following the regularity condition (\ref{eq:bounded}), 
\begin{align*}
&\mathbb{E}\Big[-\frac{\partial^2 \log f(X\vert Y,\theta)}{\partial \theta_i \partial \theta_j}\Big\vert Y=y,\theta\Big]\\
&\hspace{0cm}=-\int_{\mathcal{X}(y)} \Big[\frac{\partial^2 \log f(x\vert y,\theta)}{\partial \theta_i \partial \theta_j}\Big] f(x\vert y,\theta)\lambda(dx)\\
&\hspace{0cm}=-\int_{\mathcal{X}(y)} \frac{\partial}{\partial \theta_i}\Big[\Big[\frac{\partial \log f(x\vert y,\theta)}{\partial \theta_j}\Big]f(x\vert y,\theta)\Big] \lambda(dx)\\
&\hspace{0cm}+\int_{\mathcal{X}(y)} \Big[\frac{\partial \log f(x\vert y,\theta)}{\partial \theta_i}\Big]\Big[\frac{\partial \log f(x\vert y,\theta)}{\partial \theta_j}\Big] f(x\vert y,\theta)\lambda(dx)\\
&\hspace{0cm}= - \frac{\partial}{\partial \theta_i}\mathbb{E}\Big[ \frac{\partial \log f(X\vert Y,\theta)}{\partial \theta_j}\Big\vert Y=y,\theta\Big] \\
& + \mathbb{E}\Big[\Big(\frac{\partial \log f(X\vert Y,\theta)}{\partial \theta_i}\Big)\Big(\frac{\partial \log f(X\vert Y,\theta)}{\partial \theta_j}\Big)\Big\vert Y=y,\theta\Big],
\end{align*}
which leads to the identity on account of (\ref{eq:condscore2}). To show that $J_{x\vert y}(\theta)$ is positive definite, let $S_{x\vert y}(\theta)=\frac{\partial \log f(X\vert Y,\theta)}{\partial \theta}$. Thus, $J_{x\vert y}(\theta)=S_{x\vert y}(\theta)S_{x\vert y}(\theta)^{\top}$. Hence, for any $0\neq z\in\mathbb{R}^d$, $z^{\top}S_{x\vert y}(\theta)S_{x\vert y}(\theta)^{\top}z>0$ showing the information matrix $J_{x\vert y}(\theta)$ is positive definite.
\end{proof}

Taking expectation w.r.t to $\mathbb{P}\{\bullet \vert \theta \}$ on both sides of the last identity leads to the fact that $-\frac{\partial^2 \log f(X\vert Y, \theta)}{\partial \theta^2} $ is the observed Fisher information, see Schervish (1995). 
\begin{theorem}[\textbf{Resulting loss of information in incomplete data}]\label{theo:theo1} For any $\theta\in\Theta$ and $Y\in\mathcal{Y}$, 
\begin{align}\label{eq:infoloss}
J_x(\theta) > J_y(\theta).
\end{align}
\end{theorem}
\begin{proof} Let $J_x(\theta_i,\theta_j)$ be the $(i,j)-$element of $J_x(\theta)$, similarly defined for $J_y(\theta_i,\theta_j)$. Taking derivative of (\ref{eq:score2}) yields $-\frac{\partial^2 \log f(X\vert Y,\theta)}{\partial \theta_i\partial \theta_j}=-\frac{\partial^2 \log f_c(X\vert \theta)}{\partial \theta_i\partial \theta_j} + \frac{\partial^2 \log f_o(Y\vert\theta)}{\partial \theta_i \partial \theta_j}.$ By taking conditional expectation $\mathbb{E}\big[\bullet \big\vert Y,\theta\big]$ we obtain, 
\begin{align*}
J_{x\vert y}(\theta_i,\theta_j)=J_x(\theta_i,\theta_j) - J_y(\theta_i,\theta_j).
\end{align*}
The inequality (\ref{eq:infoloss}) follows from the fact that the information matrix $J_{x\vert y}(\theta)$ is positive definite, by Theorem \ref{theo:theomain}.
 \end{proof}

By Theorem \ref{theo:theo1} we deduce the following inequality corresponding to the resulting information loss presented in incomplete-data. The inequality was discussed in Orchard and Woodbury (1972), Blahut (1987), Theorem 2.86 of Schervish (1995), and McLachlan and Krishnan (2008). 
\begin{align}\label{eq:posdef}
I_x(\theta) > I_y(\theta)>0.
\end{align}
Note that the positive definiteness of $I_y(\theta)$ is due to the fact that it is the expected Fisher information matrix satisfying $I_y(\theta)=\mathbb{E}\big[\big(\frac{\partial \log f_o(Y\vert\theta)}{\partial \theta}\big) \big(\frac{\partial \log f_o(Y\vert\theta)}{\partial \theta}\big)^{\top}\big\vert\theta\big]>0.$

By applying the conditional probability density (\ref{eq:condpdf}) and the above identities, we derive the following result.

\begin{proposition}\label{prop:proptheo2}
For any $(\theta_i,\theta_j)\in \theta$ and $Y\in\mathcal{Y}$,
\begin{align*}
J_{x\vert y}(\theta_i,\theta_j)=& \mathbb{E}\Big[ \Big(\frac{\partial \log f_c(X\vert\theta)}{\partial \theta_i}\Big)\Big(\frac{\partial \log f_c(X\vert\theta)}{\partial \theta_j}\Big)\Big\vert Y,\theta\Big]  \nonumber\\
&\hspace{-1.5cm}-\mathbb{E}\Big[ \frac{\partial \log f_c(X\vert\theta)}{\partial \theta_i}\Big\vert Y,\theta\Big]\mathbb{E}\Big[ \frac{\partial \log f_c(X\vert\theta)}{\partial \theta_j}\Big\vert Y,\theta\Big].\nonumber
\end{align*}
\end{proposition}
\begin{proof} Using (\ref{eq:score2}), the results of Proposition \ref{prop:prop2}, Corollary \ref{cor:zerocov} and that $\mathrm{Cov}\left(\frac{\partial \log f_o(Y\vert\theta)}{\partial \theta_i},\frac{\partial \log f_o(Y\vert\theta)}{\partial \theta_j}\Big\vert Y,\theta\right)=0$,
\begin{align*}
&\mathrm{Cov}\Big(\frac{\partial \log f_c(X\vert\theta)}{\partial \theta_i},\frac{\partial \log f_c(X\vert\theta)}{\partial \theta_j}\Big\vert Y,\theta\Big)\\
&=\mathrm{Cov}\Big(\frac{\partial \log f(X\vert Y,\theta)}{\partial \theta_i},\frac{\partial \log f(X\vert Y,\theta)}{\partial \theta_j}\Big\vert Y,\theta\Big)\\
&= \mathbb{E}\Big[\frac{\partial \log f(X\vert Y,\theta)}{\partial \theta_i}\frac{\partial \log f(X\vert Y,\theta)}{\partial \theta_j}\Big\vert Y,\theta\Big],
\end{align*}
which is equal to $J_{x\vert y}(\theta_i,\theta_j)$ by Theorem \ref{theo:theomain}.
\end{proof}

The result below presents explicit form of the $(i,j)$-element $J_y(\theta_i,\theta_j)=-\frac{1}{n}\sum_{k=1}^n \frac{\partial^2\log f_o(Y^k\vert\theta)}{\partial \theta_i\partial \theta_j}$, with $(\theta_i,\theta_j)\in\theta$, of the observed Fisher information matrix $J_y(\theta)$.
\begin{theorem}\label{theo:infomat}
[(i)] The $(i,j)$-element of the observed information $J_y(\theta)$ for incomplete data $Y=\bigcup_{k=1}^n Y^k$ is  
\begin{align*}
J_y(\theta_i,\theta_j)=&\frac{1}{n}\sum_{k=1}^n \mathbb{E}\Big[-\frac{\partial^2  \log f_c(X^k\vert\theta)}{\partial \theta_i \partial \theta_j}\Big\vert Y^k,\theta\Big] \nonumber\\
&\hspace{-2cm}-\frac{1}{n}\sum_{k=1}^n \mathbb{E}\Big[ \Big(\frac{\partial \log f_c(X^k\vert\theta)}{\partial \theta_i}\Big)
\Big(\frac{\partial \log f_c(X^k\vert\theta)}{\partial \theta_j}\Big)\Big\vert Y^k,\theta\Big]  \nonumber\\
&\hspace{-2cm}+\frac{1}{n}\sum_{k=1}^n\mathbb{E}\Big[ \frac{\partial \log f_c(X^k\vert\theta)}{\partial \theta_i}\Big\vert Y^k,\theta\Big]\mathbb{E}\Big[ \frac{\partial \log f_c(X^k\vert\theta)}{\partial \theta_j}\Big\vert Y^k,\theta\Big].\nonumber
\end{align*}
[(ii)] Both $J_y(\widehat{\theta}_n)$ and $J_x(\widehat{\theta}_n)$ are positive definite with 
\begin{align}\label{eq:infoloss2}
J_x(\widehat{\theta}_n) > J_y(\widehat{\theta}_n)>0.
\end{align}
\end{theorem}
\begin{proof} Let $J_y^k(\theta_i,\theta_j)=- \frac{\partial^2\log f_o(Y^k\vert\theta)}{\partial \theta_i\partial \theta_j}$, $J_x^k(\theta_i,\theta_j)=\mathbb{E}\big[-\frac{\partial^2 \log f_c(X^k\vert\theta)}{\partial \theta_i\partial \theta_j}\big\vert Y^k,\theta\big]$. Similarly for $J_{x\vert y}^k(\theta_i,\theta_j)$. From (\ref{eq:score2}), $-\frac{\partial^2 \log f_o(Y^k\vert\theta)}{\partial \theta_i\partial \theta_j}=-\frac{\partial^2 \log f_c(X^k\vert\theta)}{\partial \theta_i\partial \theta_j} + \frac{\partial^2 \log f(X^k\vert Y^k,\theta)}{\partial \theta_i\partial \theta_j}.$ Taking conditional expectation $\mathbb{E}\big[\bullet \big\vert Y^k,\theta\big]$ we obtain
\begin{align*}
J_y^k(\theta_i,\theta_j)=J_x^k(\theta_i,\theta_j)-J_{x\vert y}^k(\theta_i,\theta_j).
\end{align*}
The proof of (i) is complete by the result of Proposition \ref{prop:proptheo2} on account that $J_y(\theta_i,\theta_j)=\frac{1}{n}\sum_{k=1}^n J_y^k(\theta_i,\theta_j)=\frac{1}{n}\sum_{k=1}^nJ_x^k(\theta_i,\theta_j)-\frac{1}{n}\sum_{k=1}^nJ_{x\vert y}^k(\theta_i,\theta_j)$. The claim (ii) is due to (\ref{eq:infoloss}) and the fact that $\widehat{\theta}_n$ maximizes (\ref{eq:criteriamle}).
\end{proof}

\begin{Rem}
Notice that the observed Fisher information $J_y(\theta)$ takes a slightly different form than Louis (1982) general matrix formula. The main differences with the latter is that it simplifies the calculation of conditional expectation of the outer product of the complete-data score function appeared in the Louis' matrix formula. And most notably, it directly verifies the asymptotic consistency of $J_y(\theta)$ to the incomplete-data Fisher information matrix $I_y(\theta)$ as the sample size $n$ increases. Also, the derivation is much simplified compared to the approach of Louis (1982) and McLachlan and Krishnan (2008).
\end{Rem}
The expression of $(i,j)-$element of information matrix $J_y(\theta)$ agrees with that of given by Frydman and Surya (2022) for the finite mixture of Markov jump processes.

\begin{corollary}\label{prop:mainprop2}
For any incomplete data $Y\in\mathcal{Y}$,
\begin{align*}
&\mathbb{E}\Big[-\frac{\partial^2 \log  f_c(X\vert\widehat{\theta}_n)}{\partial \theta^2}\Big\vert Y,\widehat{\theta}_n\Big] \nonumber \\
&\hspace{0.5cm}> \mathbb{E}\Big[ \Big(\frac{\partial \log f_c(X\vert \widehat{\theta}_n)}{\partial \theta}\Big)\Big(\frac{\partial \log f_c(X\vert\widehat{\theta}_n)}{\partial \theta}\Big)^{\top} \Big\vert Y, \widehat{\theta}_n\Big].
\end{align*}
\end{corollary}

Corollary \ref{prop:mainprop2} extends the result on equivalence between unconditional variance of complete-data score function and the expected Fisher information, see Schervish (1995). 

\subsection{Recursive algorithms for $J_y^{-1}(\widehat{\theta}_n)$ and MLE $\widehat{\theta}_n$ }

This section discusses recursive estimations of the MLE $\widehat{\theta}_n$ and the finite-sample Cr\'amer-Rao lower bound $J_y^{-1}(\widehat{\theta}_n)$ for the covariance matrix of $\widehat{\theta}_n$.

\subsubsection{\textbf{Recursive calculation of $J_y^{-1}(\widehat{\theta}_n)$}}

Hero and Fessler (1994) proposed a recursive equation for the valuation of the Cram\'er-Rao lower bound $I_y^{-1}(\theta)$, the inverse of the expected Fisher information $I_y(\theta)$. The method uses the information matrix $I_y(\theta)$ and the inverse of the expected complete-data Fisher information $I_x(\theta)$. Thus, the method avoids taking the inverse of $I_y(\theta)$ which may be more difficult to invert than $I_x(\theta)$. However, their result is not immediately applicable for incomplete data in general since the expected Fisher information matrices $I_y(\theta)$ and $I_x(\theta)$ may not be available in closed form. 

To overcome this problem, we generalize their results for the inverse of observed Fisher information $J_y^{-1}(\widehat{\theta}_n)$ based on the conditional observed information matrices $J_y(\widehat{\theta}_n)$ and $J_x(\widehat{\theta}_n)$. The key to deriving the recursive equation for the inverse $J_y^{-1}(\widehat{\theta}_n)$ is the inequality (\ref{eq:infoloss2}).
\begin{theorem}\label{theo:theo4}
Let $\{\Psi_{\ell}\}_{\ell\geq 0}$ be a sequence of $(d\times d)-$matrices, with $d=\vert \theta\vert$, and $\Psi_0=0$ satisfying
\begin{align}\label{eq:itervar}
\Psi_{\ell+1}= A(\widehat{\theta}_n) \Psi_{\ell} + B(\widehat{\theta}_n),
\end{align}
for $A(\widehat{\theta}_n),B(\widehat{\theta}_n)\in\mathbb{R}^{d\times d}$. Then, $\{\Psi_{\ell}\}_{\ell\geq 1}$ converges with root of convergence $\rho(A(\widehat{\theta}_n))$ to $\Psi=J_y^{-1}(\widehat{\theta}_n)$ with $$A(\widehat{\theta}_n)=\big[I-J_x^{-1}(\widehat{\theta}_n)J_y(\widehat{\theta}_n)\big] \;\;\textrm{and}\;\; B(\widehat{\theta}_n)=J_x^{-1}(\widehat{\theta}_n).$$
Furthermore, the convergence is monotone in the sense $$\Psi_{\ell}< \Psi_{\ell+1}\leq \Psi \quad \textrm{for $\ell=0,1,\ldots$}$$
\end{theorem}

\begin{proof} Since by Theorem \ref{theo:infomat}[(ii)] $J_y(\widehat{\theta}_n)$ and $J_x(\widehat{\theta}_n)$ are positive definite, it follows from (\ref{eq:infoloss}) and Theorem 7.2.1 on p. 438 of Horn and Johnson (2013) that $0< J_x^{-1}(\widehat{\theta}_n)J_y(\widehat{\theta}_n)<I$. Therefore, all eigenvalues of $I-J_x^{-1}(\widehat{\theta}_n)J_y(\widehat{\theta}_n)$ are positive and strictly less than one. See Corollary 1.3.4 in Horn and Johnson (2013). By Corollary 5.6.16 of Horn and Johnson (2013), it leads to
\begin{align*}
\Psi=&\big[I-J_x^{-1}(\widehat{\theta}_n)\big(J_x(\widehat{\theta}_n)-J_y(\widehat{\theta}_n)\big)\big]^{-1}J_x^{-1}(\widehat{\theta}_n)\\
=&\Big(\sum_{\ell=0}^{\infty} \big[I - J_x^{-1}(\widehat{\theta}_n)J_y(\widehat{\theta}_n)\big]^{\ell}\Big)J_x^{-1}(\widehat{\theta}_n).
\end{align*}
Since all eigenvalues of $I-J_x^{-1}(\widehat{\theta}_n)J_y(\widehat{\theta}_n)$ are positive and strictly less than one, we then obtain
\begin{align*}
\Psi_{\ell+1}- \Psi=&\big[I-J_x^{-1}(\widehat{\theta}_n)J_y(\widehat{\theta}_n)\big]\Psi_{\ell} + J_x^{-1}(\widehat{\theta}_n) -\Psi\\
=&\big[I-J_x^{-1}(\widehat{\theta}_n)J_y(\widehat{\theta}_n)\big]\big[\Psi_{\ell} -\Psi\big]\rightarrow 0,
\end{align*}
as $\ell \rightarrow \infty$, with root of convergence factor given by the maximum absolute eigenvalues of $I-J_x^{-1}(\widehat{\theta}_n)J_y(\widehat{\theta}_n)$. Thus, $\{\Psi_{\ell}\}_{\ell\geq 1}$ converges to $\Psi$ as $\ell\rightarrow \infty$. To show the convergence is monotone, recall that
\begin{align*}
\Psi_{\ell+1}-\Psi_{\ell}=&\big[I-J_x^{-1}(\widehat{\theta}_n)J_y(\widehat{\theta}_n)\big]\big[\Psi_{\ell}-\Psi_{\ell-1}\big]\\
=&\big[I-J_x^{-1}(\widehat{\theta}_n)J_y(\widehat{\theta}_n)\big]^{\ell}J_x^{-1}(\widehat{\theta}_n)>0,
\end{align*}
which is positive definite by Corollary 7.7.4(a) of Horn and Johnson (2013) and by the inequality (\ref{eq:infoloss2}). 
\end{proof}

\subsubsection{\textbf{Fisher scoring method for incomplete data}}

One of the widely used recursive scheme for parameter estimation is the Fisher scoring method, see among others Osborne (1992), Hastie et al. (2009) and Takai (2020). The scheme follows from (\ref{eq:pers1}) by replacing the true value $\theta^0$ by an estimate $\theta_{\ell}$ obtained after $\ell$ iterations and use $\theta_{\ell+1}$ in place of $\widehat{\theta}_n$. For the case of complete data the score vector $S_n(\widehat{\theta}_{\ell})$ and observed information matrix $J_y(\widehat{\theta}_{\ell})$ are defined by (\ref{eq:score}) and (\ref{eq:OFIM}), respectively. The main difficulty in dealing with incomplete data is that the information matrix $J_y(\theta)$ is difficult to evaluate. To overcome this difficulty, Takai (2020) proposed the use of lower-bound algorithm of Bohning and Linday (1988). However, although it eases the difficulty, the lower-bound does not actually correspond to the observed information matrix $J_y(\widehat{\theta}_{\ell})$. 

For this reason, we use the explicit form of information matrix $J_y(\theta)$ stated in Theorem \ref{theo:infomat}. The scheme reads
\begin{align}\label{eq:NRAlgo}
\widehat{\theta}_{\ell+1} =\widehat{\theta}_{\ell} + J_y^{-1}(\widehat{\theta}_{\ell})S_n(\widehat{\theta}_{\ell}),
\end{align}
provided the matrix $J_y(\theta)$ is invertible for any $\theta\in \Theta$. Taking account of the inequality (\ref{eq:infoloss}) and the identity (\ref{eq:id10}), the above Fisher scoring method can be improved using (\ref{eq:EMLEb}) by replacing $\theta^0$ by an estimate $\widehat{\theta}_{\ell}$ obtained after $\ell-$iteration and use $\widehat{\theta}_{\ell+1}$ in place of $\widehat{\theta}_n^0$. The recursive scheme of the estimation is discussed in the section below. 

\subsubsection{\textbf{EM-Gradient algorithm}}
The recursive equation corresponds to the \textit{EM-Gradient algorithm} proposed by Lange (1995). The recursion provides the fastest Newton-Raphson algorithm for solving the M-step iteratively which has quadratic convergence compared to the linear convergence in the EM algorithm. See Wu (1983) for the convergence properties of the EM algorithm. The EM-Gradient algorithm \eqref{eq:NRAlgo} yields an estimate $\widehat{\theta}_{\ell}$ for $\theta^0$ which serves as the lower bound to that of given by the incomplete-data Fisher scoring method, see for e.g. Osborne (1992), McLachlan and Krishnan (2008) and Takai (2020), whose estimated standard error is given by the inverse observed Fisher information matrix $J_y^{-1}(\widehat{\theta}_{\ell})$. 

\subsection*{The EM-Gradient algorithm}

\begin{enumerate}

\item[(i)][\textbf{Initial step}] Set an initial value $\widehat{\theta}_0$ of $\theta^0$. 

\item[(ii)] \textbf{E-step}, after $\ell-$th iteration, evaluate using the current estimate $\widehat{\theta}_{\ell}$ the conditional expectations 
\begin{align*}
S_n( \widehat{\theta}_{\ell}):=&\frac{1}{n}\sum_{k=1}^n\mathbb{E}\Big[ \frac{\partial \log f_c\big(X^k\vert \widehat{\theta}_{\ell}\big)}{\partial \theta}\Big\vert Y^k, \widehat{\theta}_{\ell}\Big],\\
J_x( \widehat{\theta}_{\ell}):=&\frac{1}{n}\sum_{k=1}^n\mathbb{E}\Big[-\frac{\partial^2 \log f_c\big(X^k\vert \widehat{\theta}_{\ell}\big)}{\partial \theta^2}\Big\vert Y^k,\widehat{\theta}_{\ell}\Big].
\end{align*}

\item[(iii)] \textbf{M-step}, get an update $\widehat{\theta}_{\ell+1}$ using the identity 
\begin{align}\label{eq:NRAlgo2}
\widehat{\theta}_{\ell+1}=\widehat{\theta}_{\ell} + J_x^{-1}(\widehat{\theta}_{\ell}) S_n(\widehat{\theta}_{\ell}).
\end{align}

\item[(iv)] \textbf{Stop} if $\Vert \widehat{\theta}_{\ell +1} - \widehat{\theta}_{\ell} \vert\vert <\varepsilon$, with $\varepsilon>0$. Otherwise, go back to (ii) and replace $\widehat{\theta}_{\ell}$ by $\widehat{\theta}_{\ell+1}$.
\end{enumerate}

\begin{Rem}\label{rem:rem4}
Thus, if the convergence criterion $\Vert \widehat{\theta}_{\ell+1} - \widehat{\theta}_{\ell} \vert\vert <\varepsilon$ is reached, it follows from (\ref{eq:NRAlgo}) and (\ref{eq:NRAlgo2}) that at its convergence the recursive estimator $\widehat{\theta}_{\infty}$ corresponds to the MLE $\widehat{\theta}_n$ since $\widehat{\theta}_{\infty}$ solves the equation $S_n(\widehat{\theta}_{\infty})=0.$
\end{Rem}

\subsection{Sandwich estimator of covariance matrix and asymptotic properties of the M-estimator $\widehat{\theta}_n^0$}
The $\sqrt{n}-$limiting normal distribution of the M-estimator $\widehat{\theta}_n^0$ is expressed in terms of the sandwich estimator.
\begin{Def}[\textbf{Finite-sample sandwich estimator}]
\begin{align}\label{eq:sandwich}
\Sigma_n(\widehat{\theta}_n)=J_x^{-1}(\widehat{\theta}_n)J_y(\widehat{\theta}_n)J_x^{-1}(\widehat{\theta}_n), 
\end{align}
defines a finite-sample consistent efficient sandwich estimator of covariance matrix for the MLE $\widehat{\theta}_n$.
\end{Def}
Note that \eqref{eq:sandwich} is slightly different from Huber sandwich estimator $V_n(\widehat{\theta})=J_y^{-1}(\widehat{\theta})K_n(\widehat{\theta})J_y^{-1}(\widehat{\theta})$, with the matrix $K_n(\theta)$ defined by $\frac{1}{n}\sum_{k=1}^n \Big(\frac{\partial \log f_o(Y^k\vert \theta)}{\partial \theta}\Big)\Big( \frac{\partial \log f_o(Y^k\vert \theta)}{\partial \theta}\Big)^{\top},$ for model misspecification under incomplete data. See e.g. Huber (1967), Freedman (2006), and Little and Rubin (2020). In contrast to the Huber sandwich estimator, \eqref{eq:sandwich} does not involve the inverse of $J_y(\theta)$ which in general may be more difficult to evaluate than $J_x(\theta)$. 
\begin{theorem}[\textbf{$\sqrt{n}-$consistent limiting normal distribution of the M-estimator $\widehat{\theta}_n^0$}]\label{theo:maintheo}
By Assumption (\ref{eq:bounded}),
\begin{align*}
\sqrt{n}\big(\widehat{\theta}_n^0-\theta^0\big) \stackrel{d}{\Longrightarrow} N(0,\Sigma(\theta^0)),
\end{align*}
with the asymptotic covariance matrix $\Sigma(\theta^0)$ defined by
\begin{align*}
\Sigma(\theta^0)=I_x^{-1}(\theta^0)I_y(\theta^0)I_x^{-1}(\theta^0).
\end{align*}
\end{theorem}

\begin{proof} As a consistent estimator of $\theta^0$, the asymptotic normality of the M-estimator $\widehat{\theta}_n^0$ follows from (\ref{eq:EMLE})-(\ref{eq:EMLEb}) and equality between $\mathscr{S}_n(\theta^0)$ (\ref{eq:solsol}) and $S_n(\theta^0)$ (\ref{eq:score}). Since $\sqrt{n}S_n(\theta^0)\sim N(0, I_y(\theta^0))$, application of Slutsky's lemma and CLT gives the limiting distribution of $\widehat{\theta}_n^0$. 
\end{proof}

By Assumption (\ref{eq:bounded}) and consistency of the MLE $\widehat{\theta}_n$, the asymptotic covariance $\Sigma(\theta^0)$ can be consistently estimated by the sandwich estimator $\Sigma_n(\widehat{\theta}_n)$ (\ref{eq:sandwich}). Furthermore, positive definiteness of $J_y(\widehat{\theta}_n)$ and $J_x(\widehat{\theta}_n)$, hence is invertible (see Theorem 7.2.1 on p. 438 of Horn and Johnson (2013)), implies following (\ref{eq:infoloss}) that $\Sigma_n(\widehat{\theta}_n)$ produces smaller standard errors of the M-estimator $\widehat{\theta}_n^0$ than that of the MLE $\widehat{\theta}_n$ given by the inverse observed Fisher information $J_y^{-1}(\widehat{\theta}_n),$ which corresponds to the Cram\'er-Rao lower bound for the covariance matrix of the MLE. 

\begin{proposition}\label{prop:prop4}
The sandwich estimator $\Sigma_n(\widehat{\theta}_n)$ satisfies the Loewner partial matrix ordering 
\begin{align}\label{eq:ordering2}
J_y^{-1}(\widehat{\theta}_n) > J_x^{-1}(\widehat{\theta}_n)> \Sigma_n(\widehat{\theta}_n)>0.
\end{align}
\end{proposition}

\subsection{Improved estimation by repeated sampling }\label{sec:improvedmle}

The results of Theorem \ref{theo:maintheo} and Proposition \ref{prop:prop4} shows that the M-estimator $\widehat{\theta}_n^0$ (\ref{eq:EMLEb}) provides a superior estimate of the true value $\theta^0$ as it attains standard error smaller than that of the MLE $\widehat{\theta}_n$ (\ref{eq:pers1}) given by the inverse $J_y^{-1}(\widehat{\theta}_n)$ of the observed Fisher information $J_y(\widehat{\theta}_n)$; even smaller than $J_x^{-1}(\widehat{\theta}_n)$. However, the M-estimator $\widehat{\theta}_n^0$ involves $\theta^0$. 

To construct $\widehat{\theta}_n^0$, suppose that it is possible to draw (independently) $K$ samples $Y_k=\{Y_k^{\ell}:\ell=1,\ldots,n\}$ of size $n$ each. For this purpose, one may use the Bootstrapping method. From each $k-$th subsample $Y_k$ an MLE is derived by applying iterative Fisher scoring method (\ref{eq:NRAlgo}), the EM-Gradient algorithm (\ref{eq:NRAlgo2}) or the EM algorithm, which results in $K$ independent sets of MLEs $\{\widehat{\theta}_n^{(\ell)}:\ell=1,\ldots,K\}$. Define $\widehat{\theta}:=\frac{1}{K}\sum_{\ell=1}^K \widehat{\theta}_n^{(\ell)}$. Since $\sqrt{n}(\widehat{\theta}_n^{(\ell)}-\theta^0)\sim N(0,I_y^{-1}(\theta^0))$, it follows that $\widehat{\theta}\Longrightarrow \theta^0$ as $K$ increases where the convergence occurs with probability one. Then, apply (\ref{eq:EMLEb}) to each sample $\{Y_{k}\}$ to obtain
\begin{align}\label{eq:bootstraping}
\widehat{\theta}_{nk}^0=\widehat{\theta} + J_{xk}^{-1}(\widehat{\theta}) S_{nk}(\widehat{\theta}),
\end{align}
where $S_{nk}(\theta)$ and $J_{xk}(\theta)$ are defined respectively by
\begin{align*}
S_{nk}(\theta)=&\frac{1}{n}\sum_{\ell=1}^n \frac{\partial \log f_o(Y_{k}^{\ell}\vert \theta)}{\partial \theta},\\
J_{xk}(\theta)=&-\frac{1}{n}\sum_{\ell=1}^n \mathbb{E}\Big[\frac{\partial^2 \log f_c(X_{k}^{\ell}\vert \theta)}{\partial \theta^2}\Big\vert Y_k^{\ell},\theta\Big],
\end{align*}
whilst $X_{k}^{\ell}$ is the complete observation of $Y_{k}^{\ell}$. It is worth mentioning that although $S_{nk}(\widehat{\theta}_n^{(\ell)})=0$, $\ell\in \{1,\ldots,K\}$, $S_{nk}(\widehat{\theta}) \neq 0$. Thus, by applying Lemma 2.8 and Theorem 2.7 (iii) in Van der Vaart (2000), it follows from (\ref{eq:bootstraping}) that
\begin{align}\label{eq:sqrtndist}
\sqrt{n}(\widehat{\theta}_{nk}^0-\theta^0)\sim N\big(0, \Sigma(\theta^0)\big).
\end{align}
For an estimate of $\Sigma(\theta^0)$, we use the sandwich estimator: $$\Sigma_n(\widehat{\theta})=\overline{J}_x^{-1}(\widehat{\theta})\overline{J}_y(\widehat{\theta})\overline{J}_x^{-1}(\widehat{\theta}),$$where $\overline{J}_x(\widehat{\theta})$ and $\overline{J}_y(\widehat{\theta})$ are defined respectively by
\begin{align*}
\overline{J}_x(\widehat{\theta})=\frac{1}{K}\sum_{k=1}^K J_{xk}(\widehat{\theta})\quad \textrm{and} \quad
\overline{J}_y(\widehat{\theta})=&\frac{1}{K}\sum_{k=1}^K J_{yk}(\widehat{\theta}),
\end{align*}
with $J_{yk}(\widehat{\theta})=-\frac{1}{n}\sum_{\ell=1}^n \frac{\partial^2 \log f_o(Y_{k}^{\ell}\vert \widehat{\theta})}{\partial \theta^2}$ calculated using the result of Theorem \ref{theo:infomat}. Implementation of this estimation method is discussed in more details in Section \ref{sec:sec5}.

\begin{Rem}
Notice that the arguments used above are not applicable in order to improve the standard errors of the MLE $\widehat{\theta}_n$ when applied to (\ref{eq:pers1}) to derive the estimator $$\widehat{\theta}_{nk}=\widehat{\theta} + J_{yk}^{-1}(\widehat{\theta}) S_{nk}(\widehat{\theta}),$$ with $J_{yk}(\theta)=-\frac{1}{n}\sum_{\ell=1}^n \frac{\partial^2 \log f_o(Y_{k}^{\ell}\vert \theta)}{\partial \theta^2}$. This is because $\sqrt{n}(\widehat{\theta}_{nk}-\theta^0)\sim N\big(0,I_y^{-1}(\theta^0)\big)$. However, in the case of complete information or in the absence of repeated sampling ($K=1$), the M-estimator $\widehat{\theta}_{nk}^0$ and the MLE $\widehat{\theta}_{nk}$ coincide, with estimated covariance matrix $J_y^{-1}(\widehat{\theta})$.
\end{Rem}

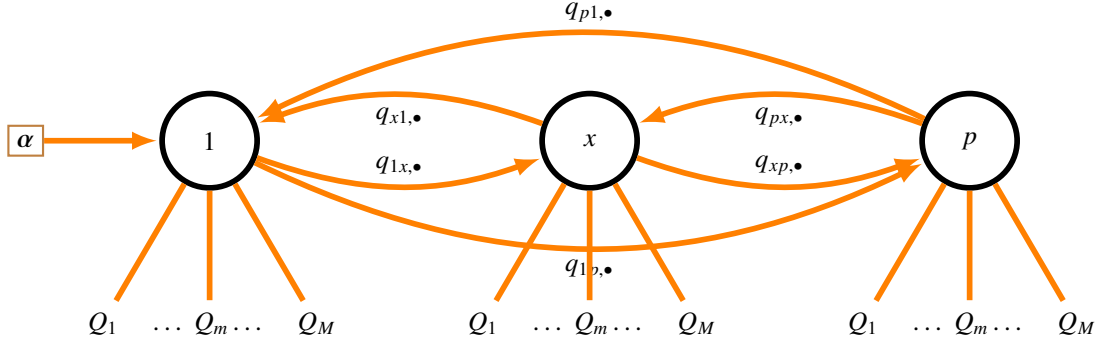
\begin{figure*}[t!]
\begin{center}
  \begin{tikzpicture}[font=\sffamily]

        \tikzset{node style/.style={state,
                                    minimum width=1.25cm,
                                    line width=0.75mm,
                                    fill=white!20!white}}

          \tikzset{My Rectangle1/.style={rectangle, draw=brown, fill=white, thick,
    prefix after command= {\pgfextra{\tikzset{every label/.style={blue}}, label=below}}
    }
}

          \tikzset{My Rectangle2/.style={rectangle,draw=brown,  fill=yellow, thick,
    prefix after command= {\pgfextra{\tikzset{every label/.style={blue}}, label=below}}
    }
}

          \tikzset{My Rectangle3/.style={rectangle, draw=brown, fill=white, thick,
    prefix after command= {\pgfextra{\tikzset{every label/.style={blue}}, label=below}}
    }
}

        \node[node style] at (2, 0)     (s1)     {$1$};
        \node[node style] at (7, 0)     (s2)     {$x$};
         \node[node style] at (12, 0)    (s3)     {$p$};

      \node [My Rectangle3, label={}] at  ([shift={(-5em,0em)}]s1.west) (p0) {$\boldsymbol{\alpha}$};

        \node [draw=none, label={} ] at  ([shift={(4em,-5em)}]s1.south) (g1) {$Q_M$};
         \node [draw=none, label={} ] at  ([shift={(0em,-5em)}]s1.south) (g10) {$\dots\; Q_m\dots$};
          \node [draw=none, label={} ] at  ([shift={(-4em,-5em)}]s1.south)  (g2) {$Q_1$};

         \node [draw=none, label={} ] at  ([shift={(4em,-5em)}]s2.south) (g3) {$Q_M$};
          \node [draw=none, label={} ] at  ([shift={(0em,-5em)}]s2.south) (g30) {$\dots\; Q_m\dots$};
          \node [draw=none, label={} ] at  ([shift={(-4em,-5em)}]s2.south)  (g4) {$Q_1$};

         \node [draw=none, label={}, auto=right ] at  ([shift={(4em,-5em)}]s3.south) (g5) {$Q_M$};
          \node [draw=none, label={} ] at  ([shift={(0em,-5em)}]s3.south) (g50) {$\dots\; Q_m\dots$};
          \node [draw=none, label={} ] at  ([shift={(-4em,-5em)}]s3.south)  (g6) {$Q_1$};

        \draw[every loop,
              auto=right,
              line width=0.75mm,
              >=latex,
              draw=orange,
              fill=orange]

            (s1)  edge[bend right=20, auto=left] node {$q_{1x,\bullet}$} (s2)

            (s2)  edge[bend right=20, auto=left]                    node {$q_{x1,\bullet}$}  (s1)

             (s3)  edge[bend right=20, auto=left]                    node {$q_{px,\bullet}$}  (s2)
             (s2)  edge[bend right=20, auto=left]                node {$q_{xp,\bullet}$}  (s3)

             (s3)  edge[bend right=26, auto=right]                node {$q_{p1,\bullet}$}  (s1)
             (s1)  edge[bend right=26, auto=right]                node {$q_{1p,\bullet}$}  (s3)

            (s1) edge [-,auto=left]  node {} (g1)
            (s1) edge [-,auto=left]  node {} (g10)
            (s1) edge [-,auto=right] node {} (g2)

            (s2) edge [-,auto=left] node {} (g3)
             (s2) edge [-,auto=left]  node {} (g30)
            (s2) edge [-,auto=right]  node {} (g4)

            (s3) edge [-,auto=left] node {} (g5)
             (s3) edge [-,auto=left]  node {} (g50)
            (s3) edge [-,auto=right] node {} (g6)

            (p0) edge node {} (s1);

 \end{tikzpicture}
 \caption{State diagram of RSCMJP process with $M$ speed regimes $\{Q_m\}$.}\label{fig:mixture}
\end{center}
\end{figure*}

\section{Regime-switching conditional Markov jump process}\label{sec:sec4}

To exemplify the results of Section \ref{sec:main}, we consider maximum likelihood estimation of the distribution parameters of a regime-switching conditional Markov jump process (RSCMJP) $X=\{X_t:t\geq 0\}$ introduced recently in Surya (2022a). RSCMJP is a complex stochastic model which can be used to describe a sequence of events where the occurrence of an event depends not only on the current state, but also on the current time and past observations of the process. It may be considered as a nontrivial generalization of the Markov jump process (see, e.g., Norris, 2009) and has distributional equivalent stochastic representation with a finite mixture of Markov jump processes proposed in Frydman and Surya (2022). See Surya (2018,2022b) for distributional properties of the finite mixture of Markov jump processes. The RSCMJP model allows the process to switch the transition rates from a finite number of transition matrices $Q_m=(q_{xy,m}:x,y\in\mathbb{S})$, $m=1,\ldots,M$ when it moves from any phase $x$ of the state space $\mathbb{S}=\{1,\ldots,p\}$, $p\in\mathbb{N}$, to another state $y\in\mathbb{S}$ with switching probability depending on the current state, time and its past information. The latter summarizes observable quantities of $X$ concerning the number of transitions $N_{xy}$ between states $(x,y)\in\mathbb{S}$, occupation time $T_x$ in each state $x\in\mathbb{S}$, and initial state indicator $B_x$ having value one if $X_0=x$, or zero otherwise. 

Figure \ref{fig:mixture} depicts the transition diagram of a RSCMJP. Beside the transition matrices $\{Q_m: m=1,\ldots,M\}$, the distribution of $X$ is characterized by an initial probability $\boldsymbol{\alpha}$ with $\alpha_x=\mathbb{P}\{X_0=x\}$ satisfying $\sum_{x=1}^p \alpha_x=1$, and regime-switching probability $\phi_{x,m}=\mathbb{P}\{X_0=X_0^{(m)}\vert X_0= x\}$ which is the probability of making an initial transition w.r.t a Markov process $X^{(m)}=\{X_t^{(m)}:t\geq 0\}$, with transition matrix $Q_m$, starting from a state $x\in\mathbb{S}$. As the underlying Markov processes $(X^{(m)}, Q_m)$, $m=1,\ldots,M$ are defined on the same state space $\mathbb{S}$, there is a hidden information $\Phi$ regarding which underlying Markov process that drives the movement of $X$ when it makes a jump from one state to another. The random variable $\Phi$ has a categorical distribution with $\mathbb{P}\{\Phi=m\}=p_m$ and $\sum_{m=1}^M p_m=1$. The pair $(X^k,\Phi^k)$ accounts for complete observation of kth paths. Define $\Phi_{k,m}=\mathbbm{1}_{\{\Phi^k=m\}}$. Note that $\sum_{m=1}^M \Phi_{k,m}=1$ for each sample path $X^k$. Consider $n$ independent paths $\{X^k:k=1,\ldots,n\}$ of $X$ (generated data or real dataset). See Surya (2022a) for more detailed algorithm on generating sample paths of RSCMJP. As distributional equivalent stochastic representation of the finite mixture of Markov jump processes, see Frydman and Surya (2022), the complete-data log-likelihood function of $(X^k,\Phi^k)$ is given by
\begin{align}\label{eq:likelihood}
\log f_c(X^k,\Phi^k\big\vert\theta)=&\sum_{m=1}^M \sum_{x=1}^p \Phi_{k,m}B_{x}^k \log \phi_{x,m} \nonumber \\&\hspace{-2.25cm}+\sum_{m=1}^M \sum_{x=1}^p\sum_{y\neq x, y=1}^p \Phi_{k,m}\Big[N_{xy}^k \log q_{xy,m} - q_{xy,m} T_{x}^k\Big],
\end{align}
where $\{\phi_{x,m}:x\in\mathbb{S}, 1\leq m\leq M\}$ are subject to the constraint $\sum_{m=1}^M \phi_{x,m}=1$ for each $x\in\mathbb{S}$. Note that $N_{xy}^k$, $T_x^k$ and $B_x^k$ are observable quantities of the $k-$th sample path $X^k$. It follows from the above log-likelihood that the proportion $\alpha_x$ can be estimated separately by $\widehat{\alpha}_x=\overline{B}_x/n$, with $\overline{B}_x:=\sum_{k=1}^n B_x^k$. Since $\widehat{\alpha}_x$ does not involve $\Phi$, it is therefore excluded from the estimation of the other parameters $\theta^0=(\phi_{x,m}^0, q_{xy,m}^0: x,y\in\mathbb{S}, 1\leq m\leq M)$. 

\subsection{\textbf{The observed information matrix $J_x(\theta)$}}

To derive the elements of information matrices $J_x(\theta)$ and $J_y(\theta)$, define $\widehat{B}_{x,m}(\theta)=\sum_{k=1}^n \widehat{\Phi}_{k,m}(\theta) B_x^k$, where  
\begin{align*}
\widehat{\Phi}_{k,m}(\theta)=\frac{f_c(X^k,\Phi^k=m \big\vert\theta)}{\sum_{m=1}^M f_c(X^k,\Phi^k=m\big\vert \theta)}.
\end{align*}
Similarly defined for $\widehat{N}_{xy,m}(\theta)$ and $\widehat{T}_{x,m}(\theta)$. Following the complete-data log-likelihood (\ref{eq:likelihood}), we obtain
\begin{align*}
&J_x(\phi_{x,m},\phi_{y,\ell})=\frac{1}{n}\sum_{k=1}^n \mathbb{E}\Big[-\frac{\partial^2 \log f_c (X^k,\Phi^k\vert \theta)}{\partial \phi_{x,m}\partial \phi_{y,\ell}}\Big\vert X^k,\theta\Big]\\
&\hspace{1cm}=\left(\frac{\widehat{B}_{y,\ell}(\theta)}{n\phi_{y,\ell}^2}\delta_m(\ell) + \frac{\widehat{B}_{y,M}(\theta)}{n\phi_{y,M}^2}\right)\delta_x(y),
\end{align*}
with $\delta_x(y)=1$ if $y=x$, or zero otherwise. Furthermore,
\begin{align*}
J_x(q_{xy,m},q_{rv,\ell})=\frac{\delta_x(r)\delta_y(v)\delta_m(\ell)}{nq_{xy,m}q_{rv,\ell}}\widehat{N}_{rv,\ell}(\theta),
\end{align*}
and
\begin{align*}
J_x(\phi_{x,m},q_{rv,\ell})=0.
\end{align*}
It follows from the matrix structure that $J_x(\theta)$ is of block diagonal and positive definite, hence is invertible. To be more precise, notice that the submatrix $J_x(\phi_{x,m},\phi_{y,\ell})$ is of block diagonal with $J_x(\phi_{x,m},\phi_{y,\ell})=0$ for $y\neq x$ and
\begin{align*}
J_x(\phi_{x,\ell},\phi_{x,m})=
\begin{cases}
\frac{\widehat{B}_{x,\ell}(\theta)}{n \phi_{x,\ell}^2} + \frac{\widehat{B}_{x,M}(\theta)}{n \phi_{x,M}^2}, & \ell=m,\\[10pt]
\frac{\widehat{B}_{x,M}(\theta)}{n \phi_{x,M}^2}, & \ell\neq m,
\end{cases}
\end{align*}
for $\ell,m=1,\ldots,M-1$. Thus, the $(M-1)\times (M-1)-$matrix $[J_x(\phi_x)]_{\ell,m}=J_x(\phi_{x,\ell},\phi_{x,m})$ reads as
\begin{align*}
J_x(\phi_x)=D_n(\phi_x) + \beta_{x,M} \mathbbm{1}\mathbbm{1}^{\top},
\end{align*}
where $D_n(\phi_x)$ is a diagonal matrix with $[D_n(\phi_x)]_{\ell,\ell}=d_{x,\ell}:=\widehat{B}_{x,\ell}(\theta)/(n \phi_{x,\ell}^2)$ and $\beta_{x,M}=\widehat{B}_{x,M}(\theta)/(n \phi_{x,M}^2)$, whilst $\mathbbm{1}$ is a $(M-1)-$vector of one. Using the Sherman and Morrison (1950) formula, the inverse of $J_x(\phi_x)$ is
\begin{align*}
J_x^{-1}(\phi_x)=D_n^{-1}(\phi_x) - \frac{\beta_{x,M} D_n^{-1}(\phi_x)\mathbbm{1}\mathbbm{1}^{\top} D_n^{-1}(\phi_x)}{1+\beta_{x,M}\mathbbm{1}^{\top}D_n^{-1}(\phi_x)\mathbbm{1}}.
\end{align*}
Furthermore, given that $D_n^{-1}(\phi_x) $ is a diagonal matrix, the $(\ell,m)-$element of the inverse $J_x^{-1}(\phi_x)$ simplifies into
\begin{align*}
[J_x^{-1}(\phi_x)]_{\ell,m}=
\begin{cases}
\frac{1}{d_{x,\ell}} - \frac{\beta_{x,M}}{d_{x,\ell}^2\left(1+\beta_{x,M}\sum_{i=1}^{M-1}\frac{1}{d_{x,i}}\right)}, & \ell=m,\\[12pt]
-\frac{\beta_{x,M}}{d_{x,\ell}d_{x,m} \left(1+\beta_{x,M}\sum_{i=1}^{M-1}\frac{1}{d_{x,i}}\right)},
 & \ell\neq m.
\end{cases}
\end{align*}
By invertibility of $J_x(\theta)$, for $\theta\in\Theta$, one can use the EM-Gradient algorithm for faster convergent estimation of $\theta^0$.

\subsection{\textbf{The observed information matrix $J_y(\theta)$}}
The elements of the observed information matrix $J_y(\theta)$ is given in Proposition 5 of Frydman and Surya (2022):
\begin{eqnarray*}
J_y(\phi_{x,m},\phi_{y,n})=\delta _{x}(y)\sum_{k=1}^{K} \frac{\widehat{\Psi }_{x,m\vert M}^k(\theta)}{\phi_{x,m}}\frac{\widehat{\Psi}_{y,\ell\vert M}^k(\theta)}{\phi_{y,\ell}} B_{y}^{k},
\end{eqnarray*}
where $\widehat{\Psi}_{x,m\vert M}^k(\theta)= \widehat{\Phi}_{k,m}(\theta) -\frac{\phi_{x,m}}{\phi_{x,M}}\widehat{\Phi}_{k,M}(\theta)$. Moreover,
\begin{align*}
&J_y(q_{xy,m},q_{rv,\ell})=\frac{\delta_x(r)\delta_y(v)\delta_m(\ell)}{q_{xy,m}q_{rv,\ell}}\widehat{N}_{rv,\ell}(\theta)\\
&\hspace{0.25cm} -\sum_{k=1}^{K}\widehat{\Phi }_{k,\ell}(\theta)\left(\delta _{m}(\ell)-\widehat{\Phi }_{k,m}(\theta)\right) \frac{\widehat{A}_{xy,m}^{k}(\theta)}{q_{xy,m}} \frac{\widehat{A}_{rv,\ell}^{k}(\theta)}{q_{rv,\ell}},
\end{align*}
with $\widehat{\Psi}_{x,m\vert M}^k(\theta)= \widehat{\Phi}_{k,m}(\theta) -\frac{\phi_{x,m}}{\phi_{x,M}}\widehat{\Phi}_{k,M}(\theta)$, and
\begin{align*}
&J_y(\phi_{x,m},q_{rv,\ell})\\
&\hspace{0.5cm}=-\sum_{k=1}^{K}\widehat{\Phi }_{k,\ell}(\theta)\Big(\delta _{m}(\ell)- 
\widehat{\Psi }_{x,m\vert M}^k(\theta)\Big)\frac{\widehat{A}_{rv,\ell}^{k}(\theta)}{q_{rv,\ell}} \frac{B_{x}^{k}}{ \phi_{x,m}}\\
&\hspace{1.5cm}+ \delta_M(\ell) \sum_{k=1}^{K}\widehat{\Phi }_{k,\ell}(\theta) \frac{\widehat{A}_{rv,\ell}^{k}(\theta)}{q_{rv,\ell}}  \frac{B_{x}^{k}}{\phi_{x,M}}.
\end{align*}

Observe that unlike $J_x(\theta)$, the observed information $J_y(\theta)$ is in general not a sparse matrix, hence is more difficult to invert than $J_x(\theta)$. In such case, we will use the iterative scheme (\ref{eq:itervar}) to get the inverse matrix $J_y^{-1}(\theta)$.

\subsection{\textbf{The M-estimator $\widehat{\theta}_n^0$ and the EM algorithm}}

On recalling that $(X^k,\Phi^k)$ constitutes a complete dataset for the $k-$th observation, the M-criterion (\ref{eq:EM}) reads
\begin{align*}
\mathscr{M}_n(\theta)=\frac{1}{n}\sum_{k=1}^n \mathbb{E}\left[\log f_c(X^k,\Phi^k\vert\theta)\big\vert X^k, \theta^0\right].
\end{align*} 
As $\widehat{\theta}_n^0$ maximizes the M-criterion $\mathscr{M}_n(\theta)$, we obtain
\begin{equation*}
\begin{split}
&\frac{1}{n}\sum_{k=1}^n \mathbb{E}\Big[\frac{\partial \log f_c(X^k,\Phi^k\vert \theta)}{\partial \phi_{x,m}}\Big\vert X^k,\theta^0\Big]\\
&\hspace{1cm}=\frac{\widehat{B}_{x,m}(\theta^0)}{n\phi_{x,m}} - \frac{\widehat{B}_{x,M}(\theta_0)}{n\phi_{x,M}}, \;\; 1\leq m\leq M-1,
\end{split}
\end{equation*}
and
\begin{equation*}
\begin{split}
&\frac{1}{n}\sum_{k=1}^n \mathbb{E}\Big[\frac{\partial \log f_c(X^k,\Phi^k\vert \theta)}{\partial q_{xy,m}}\Big\vert X^k,\theta^0\Big]\\
&\hspace{1cm}=\frac{\widehat{N}_{xy,m}(\theta^0)}{nq_{xy,m}}- \frac{\widehat{T}_{x,m}(\theta^0)}{n},\;\; 1\leq m\leq M,
\end{split}
\end{equation*}
for $x,y\in\mathbb{S}$. On account that $\widehat{\phi}_{x,M}=1-\sum_{m=1}^{M-1} \widehat{\phi}_{x,m}$ and $\widehat{B}_{x,M}(\theta^0)=\overline{B}_x - \sum_{m=1}^{M-1} \widehat{B}_{x,m}(\theta^0)$, the components $\widehat{\phi}_{x,m}^0$ and $\widehat{q}_{xy,m}^0$ of the M-estimator $\widehat{\theta}_n^0$ are given by 
\begin{align}\label{eq:mlecmjp}
\widehat{\phi}_{x,m}^0=\frac{\widehat{B}_{x,m}(\theta^0)}{\overline{B}_x} \quad \textrm{and} \quad 
\widehat{q}_{xy,m}^0=\frac{\widehat{N}_{xy,m}(\theta^0)}{\widehat{T}_{x,m}(\theta^0)}.
\end{align}
Since the true value $\theta^0$ is unknown, the components $\phi_{x,m}^0$ and $q_{xy,m}^0$ are estimated iteratively by replacing $\theta^0$ in (\ref{eq:mlecmjp}) by $\widehat{\theta}^{(\ell)}$ obtained after $\ell-$iteration and use $\widehat{\theta}^{(\ell+1)}$ in place of $\widehat{\theta}_n^0$. By doing so, (\ref{eq:mlecmjp}) yields the EM-iteration scheme, 
\begin{eqnarray}\label{eq:EMEst}
\hspace{1cm} \widehat{\phi}_{x,m}^{(\ell+1)}=\frac{\widehat{B}_{x,m}(\widehat{\theta}^{(\ell)})}{\overline{B}_x} \quad \textrm{and} \quad 
\widehat{q}_{xy,m}^{(\ell+1)}=\frac{\widehat{N}_{xy,m}(\widehat{\theta}^{(\ell)})}{\widehat{T}_{x,m}(\widehat{\theta}^{(\ell)})},
\end{eqnarray}
with $\widehat{\theta}^{(0)}$ is chosen based on the sample data. See Frydman and Surya (2022) for details. At the convergence, the recursive estimator (\ref{eq:EMEst}) coincides with the MLE $\widehat{\theta}_n$. 

\section{Simulation study}\label{sec:sec5}
This section verifies the main results of Section \ref{sec:main} using the example of conditional Markov jump process discussed in Section \ref{sec:sec4}. For this purpose, the model parameters $\theta^0=(\phi_{x,m}^0,q_{rv,\ell}^0: (x,r,v)\in\mathbb{S}, \ell,m\in\{1,\ldots, M\})$ are set to have the following values used in Surya (2022a).

\subsection{Parameter value of $\theta^0$}\label{subsec:initial}

\begin{figure}[tp!]
\centering
\includegraphics[width=.95\linewidth]{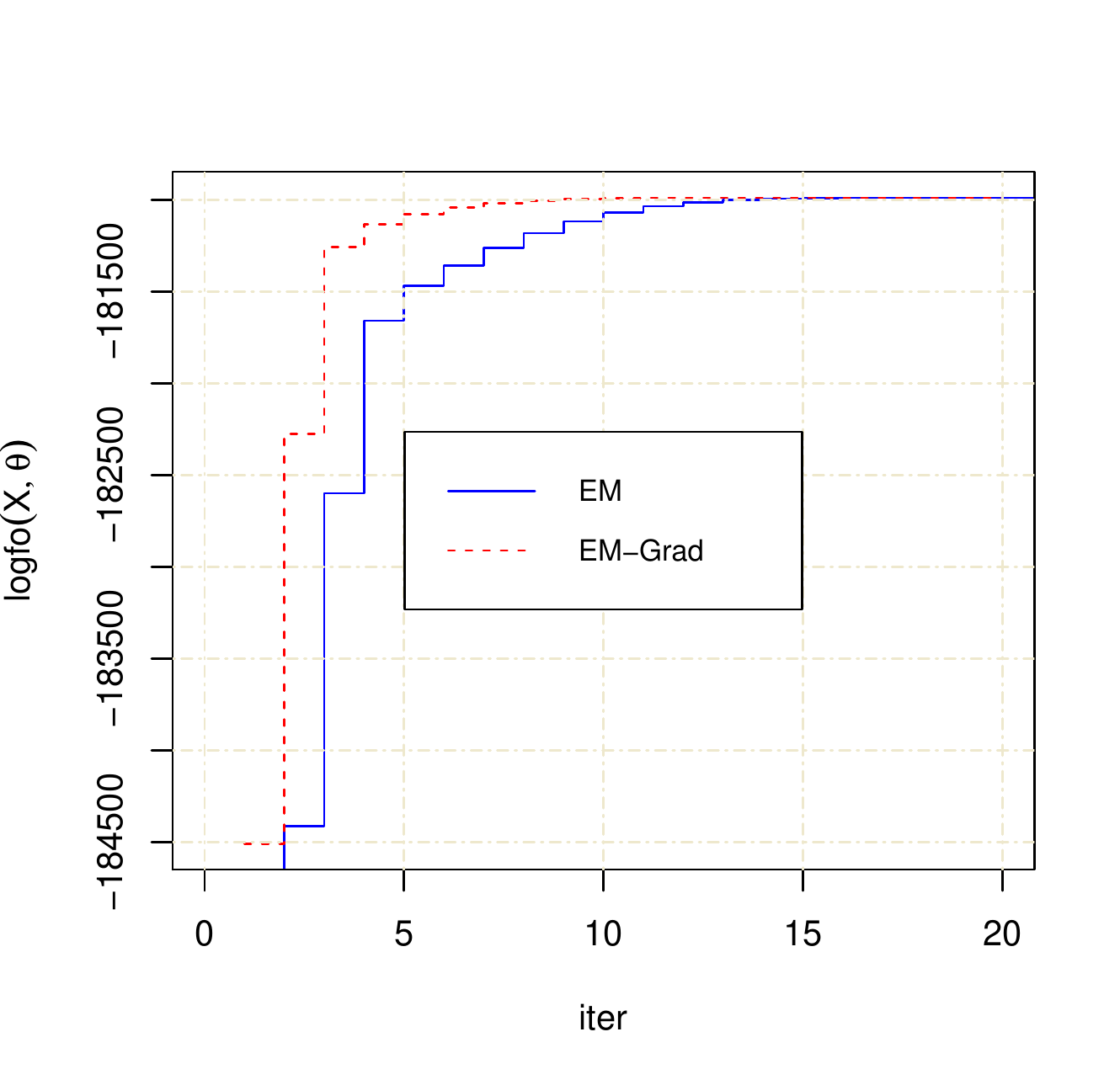}
\caption{Plot of log-likelihood $\log f_o(X_k\vert \widehat{\theta}_{nk}^{\ell})$ of the MLE obtained by EM algorithm compared to that of the EM-Gradient algorithm. } \label{fig:LogLc}
\end{figure}

\begin{figure}[tp!]
\centering
\includegraphics[width=.95\linewidth]{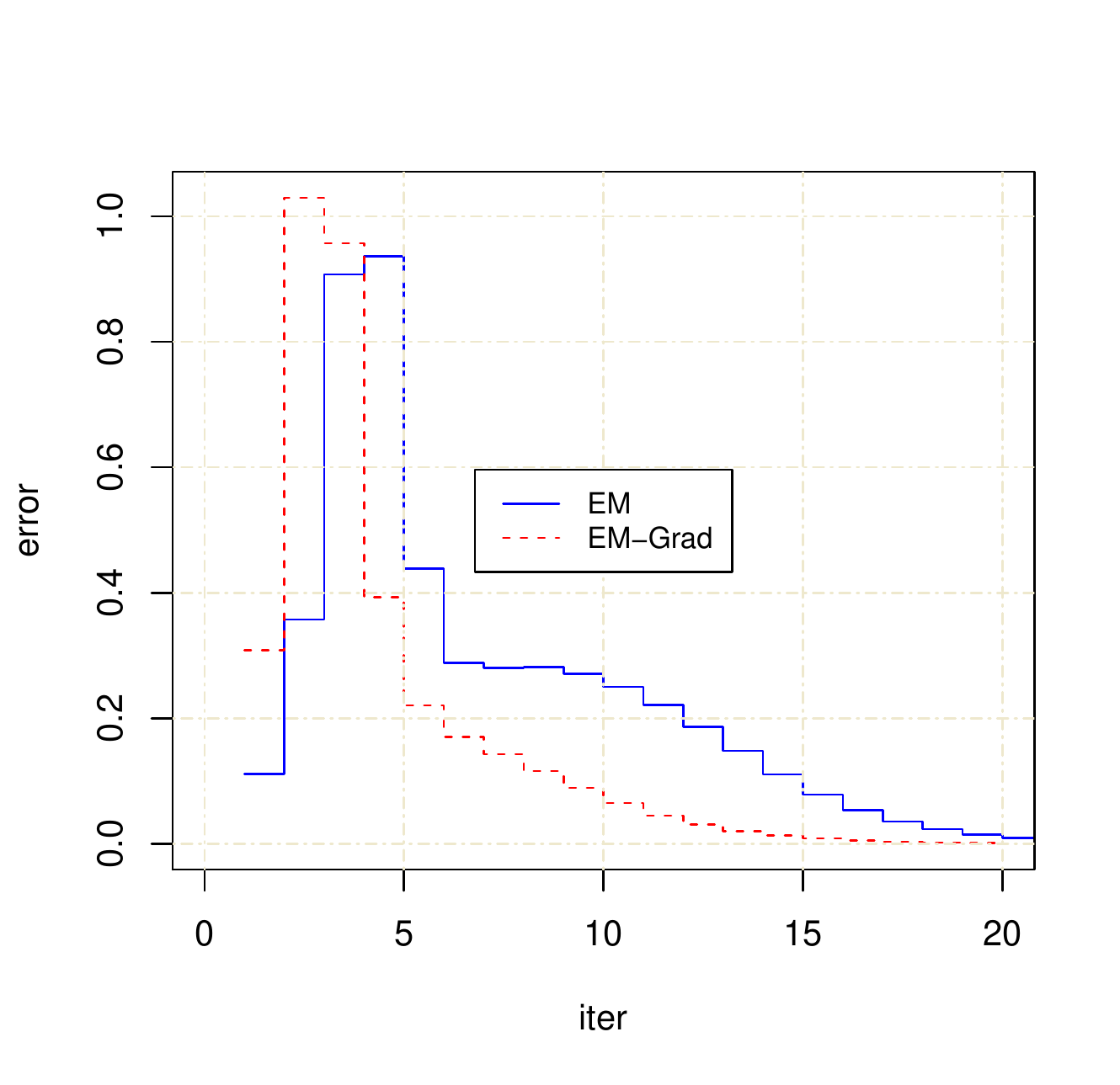}
\caption{Estimation error $\Vert\widehat{\theta}_{nk}^{\ell+1}-\widehat{\theta}_{nk}^{\ell}\Vert$ of the MLE $\widehat{\theta}_{nk}$ using the EM algorithm compared to that of the EM-Gradient algorithm. } \label{fig:Error}
\end{figure}

\begin{table}[h!]
\centering 
\begin{tabular}{ccccc}
\hline
State (x) & $\alpha_x^0$ & $\phi_{x,1}^0$ & $\phi_{x,2}^0$ & $\phi_{x,3}^0$ \\%
[0.5ex] \hline
1 & 1/3  & 0.5 & 0.3  & 0.2\\ 
2 & 1/3  & 0.25 & 0.55 & 0.2\\ 
3 & 1/3 & 0.6  & 0.1 & 0.3\\[1ex] \hline
\end{tabular}
\caption{Parameter values for $\alpha_{x}^0$ and $\phi_{x,m}^0$, $m=1,2,3$. }
\label{table:parvalue}
\end{table}

\begin{table*}[t!]
{\centering 
\par
\begin{tabular}{|l|l|l|l|l|l|l|l|l|l|l|}
\hline
\multirow{2}{*}{$\;\;\theta$} & \multicolumn{1}{c|}{True} & \multicolumn{1}{c|}{Est.} & \multicolumn{4}{c|}{Standard Errors ($\%$)}  & \multicolumn{1}{c|}{KS} \\ 
& \,\,\textrm{Value} & $\;\;\;\widehat{\theta}_n$ & $\textrm{RMSE}$ &\; $\sqrt{\frac{\overline{J}_y^{-1}}{n}}$ &\; $\sqrt{\frac{\Psi_{\ell}}{n}}$ &\;$\sqrt{\frac{\Sigma_n}{n}}$ &  \,\,\, \textrm{test} \\ 
\hline
$\phi_{1,1}$ & 0.5000 & 0.5001 & 1.9653 & 1.9331 & 1.9331 & 1.0767 & 0.8670    \\ 
$\phi_{1,2}$ & 0.3000 & 0.3011 & 1.7579 & 1.7687 & 1.7687 & 1.0084 & 0.4458  \\ 
$\phi_{2,1}$ & 0.2500 & 0.2497 & 1.7939 & 1.8184 & 1.8184 & 0.8739 & 0.0557  \\ 
$\phi_{2,2}$ & 0.5500 & 0.5502 & 1.8725 & 1.8627 & 1.8627 & 1.1420 & 0.8420  \\ 
$\phi_{3,1}$ & 0.6000 & 0.6008 & 1.7988 & 1.7581 & 1.7581 & 1.1113 & 0.7539  \\ 
$\phi_{3,2}$ & 0.1000 & 0.0995 & 1.2848 & 1.2371 & 1.2371 & 0.5909 & 0.7419  \\ 
$q_{12,1}$ & 1.2000 & 1.2006 & 1.3919 & 1.6015 & 1.6015 & 1.1406 & 0.6706 \\ 
$q_{13,1}$ & 0.8000 & 0.7998 & 1.1935 & 1.2252 & 1.2252 & 0.9338 & 0.9835  \\ 
$q_{21,1}$ & 0.2000 & 0.1999 & 0.2497 & 0.2511 & 0.2511 & 0.1964 & 0.8575  \\ 
$q_{23,1}$ & 0.2000 & 0.1996 & 0.2433 & 0.2506 & 0.2506 & 0.1965 & 0.2385  \\ 
$q_{31,1}$ & 1.2000 & 1.1984 & 1.9141 & 2.0510 & 2.0510 & 1.3234 & 0.3910  \\ 
 $q_{32,1}$ & 1.8000 & 1.7968 & 2.0537 & 2.1452 & 2.1452 & 1.7263 & 0.0755  \\ 
$q_{12,2}$ & 2.4000 & 2.3989 & 4.2487 & 4.0711 & 4.0711 & 2.4568 & 0.6955  \\ 
$q_{13,2}$ & 0.6000 & 0.5990 & 1.8485 & 1.7905 & 1.7905 & 1.2667 & 0.7836 \\ 
$q_{21,2}$ & 0.2000 & 0.1998 & 0.3052 & 0.2882 & 0.2882 & 0.2268 & 0.1171  \\ 
 $q_{23,2}$ & 0.2000 & 0.2000 & 0.2838 & 0.2884 & 0.2884 & 0.2268 & 0.7685  \\ 
$q_{31,2}$ & 0.4000 & 0.3999 & 1.3428 & 1.3536 & 1.3536 & 0.8441 & 0.6424 \\ 
$q_{32,2}$ & 1.6000 & 1.5994 & 2.2547 & 2.2802 & 2.2802 & 1.8279 & 0.5511  \\ 
$q_{12,3}$ & 1.6000 & 1.5987 & 2.5430 & 2.4654 & 2.4654 & 2.1908 & 0.3697 \\ 
$q_{13,3}$ & 2.4000 & 2.4019 & 3.3403 & 3.2976 & 3.2976 & 2.5562 & 0.9697  \\ 
$q_{21,3}$ & 0.2000 & 0.2001 & 0.3471 & 0.3200 & 0.3200 & 0.2788 & 0.8565  \\ 
$q_{23,3}$ & 0.2000 & 0.2000 & 0.2993 & 0.3198 & 0.3198 & 0.2787 & 0.6572  \\ 
$q_{31,3}$ & 3.0000 & 3.0025 & 3.6650 & 4.0338 & 4.0338 & 3.2132 & 0.1953  \\ 
$q_{32,3}$ & 2.0000 & 1.9986 & 3.1396 & 3.0636 & 3.0636 & 2.7344 & 0.8978 \\ 
     \hline
\end{tabular}
\caption{Comparison between true value $\theta^0$, the MLE $\widehat{\theta}_n$ and its standard errors using inverse of observed Fisher information $\frac{\overline{J}_y^{-1}(\widehat{\theta}_n)}{n}$, iterative estimator $\frac{\Psi_{\ell}}{n}$ \eqref{eq:itervar}, and the sandwich estimator $\frac{\Sigma_n(\widehat{\theta}_n)}{n}$ for $K=200$ independent repeated sampling of sample paths of size $n=4000$. The last column provides the p-values of the Kolmogorov-Smirnov statistics.}
\label{table:tabres1}
}
\end{table*}
\begin{table*}[t!]
{\centering 
\par
\begin{tabular}{|l|l|l|l|l|l|l|l|l|l|l|}
\hline
\multirow{2}{*}{$\;\;\theta$} & \multicolumn{1}{c|}{True} & \multicolumn{1}{c|}{Est.} & \multicolumn{4}{c|}{Standard Errors ($\%$)}  & \multicolumn{1}{c|}{KS} \\ 
& \,\,\textrm{Value} & $\;\;\;\widehat{\theta}_n^0$ & $\textrm{RMSE}$ &\;$\sqrt{\frac{\Sigma_n}{n}}$ &\; $\sqrt{\frac{\overline{J}_y^{-1}}{n}}$ &\; $\sqrt{\frac{\Psi_{\ell}}{n}}$ &  \,\,\, \textrm{test} \\ 
\hline
$\phi_{1,1}$ & 0.5000 & 0.5004 & 1.1175 & 1.0779 & 1.9388 & 1.9387 & 0.6145 \\ 
$\phi_{1,2}$ & 0.3000 & 0.3003 & 1.0836 & 1.0098 & 1.7741 & 1.7741 & 0.7321 \\ 
$\phi_{2,1}$ & 0.2500 & 0.2501 & 0.8180 & 0.8773 & 1.8214 & 1.8214 & 0.4728 \\ 
$\phi_{2,2}$ & 0.5500 & 0.5488 & 1.1116 & 1.1454 & 1.8659 & 1.8658 & 0.3376 \\ 
$\phi_{3,1}$ & 0.6000 & 0.6006 & 1.1202 & 1.1110 & 1.7662 & 1.7662 & 0.4826 \\ 
$\phi_{3,2}$ & 0.1000 & 0.1007 & 0.6369 & 0.5987 & 1.2513 & 1.2513 & 0.1117 \\ 
$q_{12,1}$ & 1.2000 & 1.2001 & 1.0790   & 1.1395 & 1.6035 & 1.6035   & 0.6775 \\ 
$q_{13,1}$ & 0.8000 & 0.8003 & 0.8610  & 0.9335  & 1.2261 & 1.2261   & 0.6608 \\ 
$q_{21,1}$ & 0.2000 & 0.1998 & 0.1957  & 0.1962 & 0.2513 & 0.2513  & 0.3156 \\ 
$q_{23,1}$ & 0.2000 & 0.1999 & 0.1989 & 0.1963 & 0.2508 & 0.2508  & 0.6091 \\ 
$q_{31,1}$ & 1.2000 & 1.1991 & 1.3643 & 1.3217 & 2.0585 & 2.0585  & 0.8254 \\ 
$q_{32,1}$ & 1.8000 & 1.7982 & 1.6917 & 1.7252 & 2.1467 & 2.1467  & 0.2406 \\ 
$q_{12,2}$ & 2.4000 & 2.3979 & 2.4644 & 2.4559 & 4.0896 & 4.0896  & 0.7365 \\ 
$q_{13,2}$ & 0.6000 & 0.5989 & 1.2403 & 1.2670 & 1.7960 & 1.7960 & 0.5133 \\ 
$q_{21,2}$ & 0.2000 & 0.2000 & 0.2261 & 0.2269 & 0.2889 & 0.2889  & 0.5564 \\ 
$q_{23,2}$ & 0.2000 & 0.2002 & 0.2249 & 0.2269 & 0.2892 & 0.2892 & 0.0983 \\ 
$q_{31,2}$ & 0.4000 & 0.3999 & 0.8562 & 0.8450 & 1.3570 & 1.3570 & 0.6364 \\ 
$q_{32,2}$ & 1.6000 & 1.6003 & 1.7598 & 1.8286 & 2.2857 & 2.2857 & 0.8173 \\ 
$q_{12,3}$ & 1.6000 & 1.6005 & 2.1762 & 2.1921 & 2.4698 & 2.4698 & 0.8779 \\ 
$q_{13,3}$ & 2.4000 & 2.4017 & 2.6142 & 2.5553 & 3.3039 & 3.3039  & 0.4707 \\ 
$q_{21,3}$ & 0.2000 & 0.1998 & 0.3005 & 0.2788 & 0.3205 & 0.3205 & 0.3840 \\ 
$q_{23,3}$ & 0.2000 & 0.1999 & 0.2668 & 0.2785 & 0.3199 & 0.3199 & 0.8889 \\ 
$q_{31,3}$ & 3.0000 & 3.0006 & 3.1549 & 3.2128 & 4.0344 & 4.0344  & 0.7093 \\ 
$q_{32,3}$ & 2.0000 & 2.0008 & 2.7380 & 2.7344 & 3.0657 & 3.0657 & 0.4681 \\ 
             \hline
\end{tabular}
\caption{Comparison between true value $\theta^0$, the M-estimator $\widehat{\theta}_n^0$ and its standard errors using inverse of observed Fisher information $\frac{\overline{J}_y^{-1}(\widehat{\theta}_n^0)}{n}$, iterative estimator $\frac{\Psi_{\ell}}{n}$ \eqref{eq:itervar}, and the sandwich estimator $\frac{\Sigma_n(\widehat{\theta}_n^0)}{n}$ for $K=200$ independent repeated sampling of sample paths of size $n=4000$. The last column provides the p-values of the Kolmogorov-Smirnov statistics.}
\label{table:tabres2}
}
\end{table*}

For simulation study, let $\mathbb{S}=\{1,2,3\}$ and $M=3$. The value of initial probabilities $\alpha_x^0$ and $\phi _{x,m}^0$ are presented in Table \ref{table:parvalue}. The intensity matrices $Q_{1}^0$, $Q_{2}^0$, and $Q_3^0$ for the regime membership $X^{(1)}$, $X^{(2)}$, and $X^{(3)}$ are 
\begin{eqnarray*}
Q_{1}^0=\left( 
\begin{array}{ccc}
-2.0 & 1.2 & 0.8 \\ 
0.2 & -0.4 & 0.2 \\ 
1.2 & 1.8 & -3.0%
\end{array}\right), 
\quad Q_{2}^0=\left( 
\begin{array}{ccc}
-3.0 & 2.4 & 0.6 \\ 
0.2 & -0.4 & 0.2 \\ 
0.4 & 1.6 & -2.0%
\end{array}%
\right),
\end{eqnarray*}%
and
\begin{eqnarray*}
Q_{3}^0=\left( 
\begin{array}{ccc}
-4.0 & 1.6 & 2.4 \\ 
0.2 & -0.4 & 0.2 \\ 
3.0 & 2.0 & -5.0%
\end{array}%
\right),
\end{eqnarray*}
respectively. We see following intensity matrices $Q_m^0$, $m=1,2,3$, that each regime $X^{(m)}$ has different expected state occupation time and the probability of making an immediate jump from one state to another, except for the transition from state 2. In the latter case, it is difficult to identify which underlying Markov jump process that drives the dynamics of $X$ when it moves out of state 2. 

\subsection{Simulation and estimation results}

\subsubsection{\textbf{Maximum likelihood estimation $\widehat{\theta}_n$}}\label{sec:sec521}

Based on the above parameters, a set of $K=200$ independent RSCMJP sample paths $X_k=\{X_k^{\ell}:\ell=1,\ldots,n\}$, $k\in\{1,\ldots,K\}$, of size $n=4000$ each are generated. See Surya (2022a) for the algorithm of generating sample paths of RSCMJP. To each set $X_k$ of sample paths, the MLE $\widehat{\theta}_{nk}$ is found using the EM algorithm (\ref{eq:EMEst}) and the EM-Gradient algorithm (\ref{eq:NRAlgo2}). In carrying out the statistical computation, the \textbf{R} language (2013) was used. 

Figure \ref{fig:LogLc} compares in each iteration the value of incomplete-data log-likelihood $\log f_o(X_k\vert \widehat{\theta}_{nk}^{\ell})$ of a randomly chosen set $X_k$ of sample paths as a function of the current estimate $\widehat{\theta}_{nk}^{\ell}$ obtained using the EM algorithm \eqref{eq:EMEst} against that of derived using the EM-Gradient algorithm \eqref{eq:NRAlgo2}. We observe that the EM-Gradient algorithm reaches its convergence faster than the EM algorithm as it requires less iterations to converge. Figure \ref{fig:Error} shows estimation error $\Vert\widehat{\theta}_{nk}^{\ell+1}-\widehat{\theta}_{nk}^{\ell}\Vert$ of the two algorithms. As we can see from the figure that the EM-Gradient algorithm is able to correct estimation error to reach its convergence faster than the EM algorithm. The explanation for this observation is due to the fact that the EM-algorithm is equivalent to the Fisher scoring iteration \eqref{eq:NRAlgo} whose successive estimation error $\Vert\widehat{\theta}_{nk}^{\ell+1}-\widehat{\theta}_{nk}^{\ell}\Vert$ is determined by the inverse $J_y^{-1}(\widehat{\theta}_{nk}^{\ell})$, whereas for the EM-Gradient by $J_x^{-1}(\widehat{\theta}_{nk}^{\ell})$, which is smaller than $J_y^{-1}(\widehat{\theta}_{nk}^{\ell})$ in the sense of Loewner partial matrix ordering (\ref{eq:ordering2}), attributed by the resulting loss of information presented in incomplete data.

To verify the asymptotic properties of the MLEs $\widehat{\theta}_{nk}$, the mean squared error $\textrm{MSE}=\frac{1}{K}\sum_{k=1}^K (\widehat{\theta}_{nk}-\theta^0)^2$ is calculated. Subsequently, the information matrices $J_{yk}:=J_{yk}(\widehat{\theta}_{nk})$ and $J_{xk}:=J_{xk}(\widehat{\theta}_{nk})$ are evaluated. Then each set of matrices $\{J_{yk}: k=1,\ldots, K\}$ and $\{J_{xk}: k=1,\ldots, K\}$ are averaged to obtain $\overline{J}_y=\frac{1}{K}\sum_{k=1}^K J_{yk}$ and $\overline{J}_x=\frac{1}{K}\sum_{k=1}^K J_{xk}$. The inverse $\overline{J}_y^{-1}$ of the averaged Fisher information $\overline{J}_y$ is used to get estimated standard error of the MLE $\widehat{\theta}_n=\frac{1}{K}\sum_{k=1}^K \widehat{\theta}_{nk}$. Also, the iterative estimator $\Psi_{\ell}$ \eqref{eq:itervar} of the Cram\'er-Rao lower bound $J_y^{-1}(\widehat{\theta}_n)$ is used for $\ell=50$ iteration steps to calculate the standard error of the MLE $\widehat{\theta}_n$. They are both compared to the sandwich estimator $\Sigma_n=\overline{J}_x^{-1}\overline{J}_y \overline{J}_x^{-1}$. Note that in the computation of $\Psi_{\ell}$, the matrices $\overline{J}_x$ and $\overline{J}_y$ are used in place of $J_x(\widehat{\theta}_n)$ and $J_y(\widehat{\theta}_n)$, respectively. Furthermore, since we consider the standard error of $\widehat{\theta}_n$, all covariance matrices $\overline{J}_y^{-1}$, $\Psi_{\ell}$ and $\Sigma_n$ are standardized by the sample size $n$.

Table \ref{table:tabres1} presents the estimation results for the MLE $\widehat{\theta}_n$, the standard errors and the Kolmogorov-Smirnov (KS) goodness-of-fit for the asymptotic distribution of the MLEs $\widehat{\theta}_{nk}$. The table shows the convergence of $\widehat{\theta}_n$ to the true value $\theta^0$ and of the iterative estimator $\Psi_{\ell}$ \eqref{eq:itervar} to $\overline{J}_y^{-1}$ (see Theorem \ref{theo:theo4}), where the iteration is evaluated at the MLE $\widehat{\theta}_n$. Also, the table shows that the root of MSEs are much closer to the standard errors produced by $\overline{J}_y^{-1}$ and are larger than those provided by $\Sigma_n$, by the matrix ordering (\ref{eq:ordering2}). By the adherence of $\widehat{\theta}_n$ to $\theta^0$, the MSEs are close to the sample variance of $\widehat{\theta}_{nk}$. The p-values of the KS statistic for the standardized error $z_{nk}:=(\widehat{\theta}_{nk}-\theta^0)/\widehat{\sigma}_{\widehat{\theta}_{nk}}$, where $\widehat{\sigma}_{\widehat{\theta}_{nk}}$ is the estimated standard deviation of $\widehat{\theta}_{nk}$, given by the corresponding diagonal element of $\overline{J}_y^{-1}$, is larger than $\alpha=5\%$ confirming the asymptotic normal distribution of the MLE $\widehat{\theta}_n$ as stated in Section \ref{sec:sec22}.

\subsubsection{\textbf{Improved estimation by the M-estimator $\widehat{\theta}_n^0$}}\label{sec:sec522}

To improve the estimation and reduce variability in the standard errors, we consider the M-estimator $\widehat{\theta}_{nk}^0$ defined as the maximizer of the M-criterion $\mathscr{M}_n(\theta):=\frac{1}{n}\sum_{\ell=1}^n \mathbb{E}\big[\log f_c(X_k^{\ell},\Phi_k^{\ell}\vert \theta) \big\vert X_k^{\ell}, \widehat{\theta}_n\big]$ as discussed in Section \ref{sec:improvedmle}. The components $\widehat{\phi}_{x,m}^0$ and $\widehat{q}_{xy,m}$ of $\widehat{\theta}_{nk}^0$ are obtained by applying each randomly generated set $X_k=\{X_k^{\ell}:\ell=1,\ldots,n\}$, $k\in\{1,\ldots,K\}$, of sample paths to the estimator (\ref{eq:mlecmjp}) by replacing the true value $\theta^0$ by the MLE $\widehat{\theta}_n$. The estimator $\widehat{\theta}_n^0=\frac{1}{K}\sum_{k=1}^K \widehat{\theta}_{nk}^0$ is reserved as the estimate of $\theta^0$. The mean squared error $\frac{1}{K}\sum_{k=1}^K (\widehat{\theta}_{nk}^0 -\theta^0)^2$ and the observed information matrices $J_{xk}:=J_{xk}(\widehat{\theta}_n)$ and $J_{yk}:=J_{yk}(\widehat{\theta}_n)$ are evaluated based on $\widehat{\theta}_n$. The standard errors from MSE are compared to those provided by the inverse $\overline{J}_y^{-1}$ of the averaged observed information matrix $\overline{J}_y=\frac{1}{K}\sum_{k=1}^K J_{yk}$, the iterative estimator $\Psi_{\ell}$ \eqref{eq:itervar} (for $\ell=50$ iterations) and the sandwich estimator $\Sigma_n=\overline{J}_x^{-1}\overline{J}_y\overline{J}_x^{-1}$ with $\overline{J}_x=\frac{1}{K}\sum_{k=1}^K J_{xk}$. As before, in the computation of $\Psi_{\ell}$, the matrices $\overline{J}_x$ and $\overline{J}_y$ are used in place of $J_x(\widehat{\theta}_n)$ and $J_y(\widehat{\theta}_n)$, respectively. 

The results are presented in Table \ref{table:tabres2}. In contrast to the results presented in Table \ref{table:tabres1}, the M-estimator $\widehat{\theta}_n^0$ shows some improvements on the estimation in terms of overall being closer to the true value $\theta^0$ than the MLE $\widehat{\theta}_n$ to $\theta^0$ with smaller standard errors (RMSEs) than RMSEs of $\widehat{\theta}_n$, given by $\overline{J}_y^{-1}(\widehat{\theta}_n)$. The table shows  the convergence of $\Psi_{\ell}$ to $\overline{J}_y^{-1}(\widehat{\theta}_n)$, confirming the result of Theorem \ref{theo:theo4}. The convergence of $\widehat{\theta}_n^0$ to $\theta^0$ makes the RMSEs close to the sample variance of $\widehat{\theta}_{nk}^0$. We observe that RMSEs are much closer to the standard errors produced by the sandwich estimator $\Sigma_n$ than those given by $\overline{J}_y^{-1}(\widehat{\theta}_n)$. The last column corresponds to the p-values of the KS statistics. All p-values are larger than $\alpha=5\%$ confirming the asymptotic normal distribution of the M-estimator $\widehat{\theta}_n^0$ with finite-sample covariance matrix $\Sigma_n$ as stated in Theorem \ref{theo:maintheo}, in particular (\ref{eq:sqrtndist}) for the repeated sampling method.

%
%
%
 
\section{Concluding Remarks}\label{sec:sec6}
This paper developed some new results on maximum likelihood estimation from incomplete data. The novelty of the approach is based on the Halmos and Savage (1949) work on the theory of sufficient statistics. The results have largely remained unexamined as far as the literatures on incomplete data are concerned, see e.g., Schervish (1995), McLachlan and Krishnan (2008), and Little and Rubin (2020). Conditional observed information matrices are introduced and utilized to their greater extent to which their finite-sample properties are established. They possess the same Loewner partial matrix ordering properties as the expected Fisher information matrices do. In particular, we show the resulting loss of information holds for finite sample of incomplete data. In its new form, the observed Fisher information of incomplete data simplifies Louis (1982) formula for the same matrix in terms of simplifying the conditional expectation of the outer product of the complete-data score function. Also, it directly verifies asymptotic consistency of the matrix. Our method of derivation of the matrix is direct and much simplified compared to those given by Louis (1982) and McLachlan and Krishnan (2008). To avoid getting an incorrect inverse of the observed information matrix, which may be attributed by the lack of sparsity and large size of the matrix, a monotone convergence recursive equation for estimator of the inverse matrix is developed. An improved estimation by the M-estimator is proposed using repeated sampling method. A consistent sandwich estimator of covariance matrix of the M-estimator is proposed using the conditional information matrices. It extends the Huber sandwich estimator (Huber, 1967; Freedman, 2006; Little and Rubin, 2020) to model misspecification under incomplete data. The main appealing feature of the proposed sandwich estimator is that unlike its counterpart, it does not involve the inverse of observed Fisher information matrix. The standard errors of the M-estimator produced by the sandwich estimator are smaller than those of the MLE given by the inverse of the observed information matrix. The difference is attributed to the resulting information loss presented in the incomplete data. The simulation study on a complex stochastic model of regime-switching conditional Markov jump processes confirms the results presented in this paper. We believe that the results should offer potential for variety of applications for maximum likelihood estimation from incomplete data.

\section*{Acknowledgments}
Part of this work was carried out while Budhi Surya was visiting the Department of Technology, Operations and Statistics of New York University Stern School of Business in September 2019. He thanks Professor Halina Frydman for the invitation and for the hospitality provided during his stay at the NYU Stern. 

\end{document}